\documentclass{article}

\usepackage[nonatbib,preprint]{neurips_2020}

\usepackage{tikz}
\usepackage{gnuplot-lua-tikz}
\usepackage{mathtools}
\usepackage{array}
\usepackage{amsmath,amssymb,bbm}
\usepackage{color}
\usepackage{bm}
\usepackage[
	lambda,
	operators,
	advantage,
	sets,
	adversary,
	landau,
	probability,
	notions,
	logic,
	ff,
	mm,
	primitives,
	events,
	complexity,
	asymptotics,
	keys]{cryptocode}
\usepackage[shortlabels]{enumitem}
\usepackage{framed}
\usepackage{mdframed}
\usepackage{booktabs}
\usepackage{nicefrac}
\usepackage{comment}

\usepackage{dutchcal}
\usepackage[square,numbers]{natbib}
\newtheorem{definition}{Definition}
\newtheorem{theorem}{Theorem}
\newtheorem{proof}{Proof}

\title{Adversarial Examples and Metrics}

\author{%
  Nico D\"ottling$^\dagger$, Kathrin Grosse$^\dagger$, Michael Backes$^\dagger$, Ian Molloy$^\ddagger$\\
  $^\dagger$ CISPA Helmholtz Center for Information Security\\
  $^\ddagger$ IBM T. J. Watson Research Center\\
  \{doettling, kathrin.grosse, backes\}@cispa.saarland, \\
  molloyim@us.ibm.com \\
}

\def \ms{\mathsf}

\def \point{\mathsf{\mathbf{x}}}  
\def \OtherPoint{\mathsf{\mathbf{z}}} 
\def \YetAnotherPoint{\overline{\mathbf{z}}} 

\newcommand{\feaNum}[0]{n} 

\newcommand{\learner}[0]{\mathsf{\mathbf{L}}} 

\newcommand{\classifier}[0]{\mathsf{\mathbf{C}}}  
\newcommand{\metric}{\mathsf{\Delta}} 

\newcommand{\weight}{\mathsf{\mathbf{w}}}
\newcommand{\weightElem}{\mathsf{w}} 
\newcommand{\mes}{\mathsf{m}}
\newcommand{\targetMSet}{\mathsf{T}} 
\def \bin{\{0,1\}}

\def \key{\ms{k}} 
\def \ciph{\ms{c}} 
\def \secp{\lambda}

\def \Hash{\ms{H}}

\def \Setup{\ms{Setup}} 
\def \Gen{\ms{Gen}}
\def \KeyGen{\ms{KeyGen}}
\def \Rec{\mathsf{Rec}}
\def \Ext{\ms{Ext}}
\def \Enc{\ms{Enc}}
\def \Dec{\ms{Dec}}

\def \paramEll{\ms{\ell}} 
\def \paramk{\ms{k}} 

\def \A{\mathcal{A}}

\def \pick{\leftarrow_\$}

\def \Hyb{\mathcal{H}}

\def \com{\mathsf{com}}

\def \witness{\mathsf{w}}

\def \Exp{\ms{Exp}}

\begin{document}
\maketitle
\begin{abstract}%
Adversarial examples are a type of attack on machine learning (ML) systems which cause misclassification of inputs. Achieving robustness against adversarial examples is crucial to apply ML in the real world. While most prior work on adversarial examples is empirical, a recent line of work establishes fundamental limitations of robust classification based on cryptographic hardness. Most positive and negative results in this field however assume that there is a fixed target metric which constrains the adversary, and we argue that this is often an unrealistic assumption.
In this work we study the limitations of robust classification if the target metric is uncertain. Concretely, we construct a classification problem, which admits robust classification by a small classifier if the target metric is known at the time the model is trained, but for which robust classification is impossible for small classifiers if the target metric is chosen after the fact. In the process, we explore a novel connection between hardness of robust classification and bounded storage model cryptography.

\end{abstract}

\section{Introduction}

\def \model{h}
\def \prob{\Pi}
\def \ECC{\mathbf{ECC}}
\def \Adv{\mathsf{Adv}}

A recent line of works~\cite{biggio2018wild,8406613,athalye2018obfuscated,carlini2017adversarial,chen2019towards,cohen2019certified,he2017adversarial,jia2020certified,lee2019tight,levine2020wasserstein,salman2019provably,sharma2017attacking,tramer2020adaptive} 
studies a class of attacks on machine learning systems commonly known as \emph{adversarial examples} or \emph{evasion attacks}.
Such attacks target a classifier $\classifier$ trained on problem $\prob$.
Assume for simplicity that $\prob$ has just two classes, 0 and 1.
When running an evasion attack $\A$, take a sample $\point$, say of class 0, and apply \emph{a small perturbation} to $\point$ yielding the adversarial example $\tilde{\point}$. The attack succeeds against $\classifier$ if $\classifier$ determines $\tilde{\point}$ to be in class $1$. In this case, we say $\A$ has fooled $\classifier$. On the other hand, we say that $\classifier$ robustly classifies the problem $\prob$ if any such evasion attack $\A$ fools $\classifier$ only with small probability.

\textbf{Choosing Metrics for Adversarial Examples} However, a critical aspect in modeling evasion attacks is how to define \emph{small perturbations}---for classification tasks that involve data one might require that a small perturbation should not be noticeable to a human observer. Making such a requirement formal can be somewhat tricky. We cannot allow the evasion adversary to perturb instances arbitrarily: If the adversary is allowed to replace instances of class $c$ with a well-formed instance of another class $c'$, then any good classifier has to determine as class $c'$, and therefore robust classification becomes an ill-defined task. Thus, we actually need to \emph{constrain} the adversary. Motivated by practical considerations, the go-to way of constraining the adversary is by defining a metric $\metric$ on the instance-space and assigning the adversary a \emph{perturbation budget} in this metric. For the example of image-classification, this metric may be something like the euclidean metric on vectors representing the images.
However, note that there is something arbitrary about fixing a metric such as the euclidean metric, as it is not clear that this metric captures all \emph{perceptible changes}. Indeed understanding human perception~\cite{glass1976pattern} and defining aligned metrics~\cite{DBLP:journals/tip/WangBSS04,zhang2018unreasonable} are still an open research questions.
Possibly also related to this question, there is an ongoing arms-race when it comes to defend adversarial examples~\citep{athalye2018obfuscated,carlini2017adversarial,he2017adversarial,sharma2017attacking,tramer2020adaptive}. 

\paragraph{General Impossibility of Robust Classification}

This motivates the question whether \emph{efficient} robust classification is possible for any classification task and any reasonable metric, or whether there are problems for which efficient robust classification is impossible in principle. 
Several recent works~\citep{bubeck2018adversarial,DBLP:conf/colt/DegwekarNV19,mahloujifar2018can} have demonstrated that under some mild cryptographic assumptions the latter is the case: there is no silver bullet against evasion attacks. 
The problems constructed by these works can be outlined as follows. 
Take an unlearnable problem $\prob$ with classes $0$ and $1$ and an instance space $\mathcal{X} \subseteq \{0,1\}^n$. 
Such problems can for example be constructed from pseudorandom functions~\cite{GoldreichGM84}. 
As target metric consider the standard hamming metric. 
Now define an error correcting code $\ECC: \bin^n \to \bin^m$ that is  efficiently encodable and decodable and can correct a few bit-errors in the hamming metric. 
Consider the following problem $\prob'$: Instances $\point'$ for class $c \in \bin$ are obtained by generating an instance $\point$ for class $c$ of problem $\prob$ and setting $\point' = (\ECC(\point),c)$. 
That is, $\point'$ consists of an error correcting encoding of $\point$, and the last bit 
of $\point$ is identical to the class $c$. Non-robust classification of $\prob'$ is easy: A non-robust classifier can decide solely on the bit $c$, which is included in the instance $\point'$. On the other hand, we can show that efficient robust classification of $\prob'$ is impossible. Consider an evasion adversary which just flips the last bit of $\point'$, i.e. the bit which signals the class of the instance. Call the perturbed instance $\tilde{\point}'$ and note that $\metric(\point',\tilde{\point}') = 1$, i.e. their distance is 1 (and therefore small) in the hamming metric. Unlike the non-robust classifier a robust classifier cannot rely on the last bit of $\tilde{\point}'$ to classify the instance, but needs to resort to $\ECC(\point)$ to classify the instance. 
Consequently, it can be shown that an \emph{efficient} robust classifier for $\prob'$ immediately yields an efficient classifier for $\prob$, which however contradicts the unlearnability of $\prob$. We can conclude that there exists no \emph{efficient} robust classifier for $\prob'$. On the other hand, note that robust classification of $\prob'$ is well-defined as there exists an \emph{inefficient} robust classifier, which decodes $\ECC(\point)$ and classifies based on $\point$.

\paragraph{Robust Classification via Randomized Smoothing}

On the other hand, many natural classification tasks are robust against a moderate amount of \emph{random noise}~\citep{shalev2014understanding}. Recently, \cite{cohen2019certified} demonstrated that robustness against random noise can be leveraged to achieve a certain amount of provable robustness against adversarial perturbations. The idea of this approach is to add additional random noise on a given sample before classification. Given that the amount of random noise is sufficiently large, any adversarial perturbation is \emph{smoothed out} by the random noise term. Follow up work has extended this approach to more noise distributions~\citep{lee2019tight,levine2020wasserstein}, to community detection tasks in graphs~\citep{jia2020certified}, and empirically improved the observed bounds via adversarial training \citep{salman2019provably}. Thus, while on the one side we know that there exists problems which in principle do not admit robust classification, a fairly natural class of problems can be robustly classified, namely problems which admit noise-tolerant classification \emph{in the same metric}.

\paragraph{The Choice of the Metric}

All of the above works establishing both positive and negative results have one aspect in common: The metric which constrains the perturbation adversary is both fixed and publicly know. 
However, in practice, the exact metric which characterizes the adversary's budget is not precisely known. 
As an example consider again perturbations of images which that are undetectable for a human observer. 
Such a perturbation might consist of modifying a few pixels locally, but also shifting the image by a small amount or rotating it slightly.
Consequently, there is no single metric which captures this kind of perturbation exactly. 
Leaving aside the question if the correct metric can be characterized 
or if it will be known at training time,
 we need to account for \emph{uncertainty in the choice of the target metric}. Several recent works have demonstrated that this is not just a hypothetical concern. Sharma and Chen~\cite{sharma2017attacking} for example demonstrated that changing the metric of a defense will break it.
 In general, extending an attack to a new metric is relatively straight-forward~\cite{8845708,wong2019wasserstein,Xiao:2018up}. This motivates the following question:
\begin{center}
\emph{What are the principal limitations in achieving robust classification if the choice of the correct target metric is uncertain?}
\end{center}

\subsection{Our Contributions}

\def \task{\mathcal{T}}
\def \class{\mathcal{C}}

In this work, we initiate the systematic study of this question. 
That is, motivated by the examples and empirical attacks above, in which determining the \emph{true} target metric constitutes a somewhat ill-posed problem, we investigate how uncertainty about the \emph{proper} metric affects the feasibility of robust classification.
More concretely, we will consider a setting in which there is not just a single metric which constrains the budget of the adversary, but rather some metric in an entire class $\class$ of potential target metrics. 
We investigate whether there are learning tasks for which simultaneously
\begin{enumerate}
    \item Robust classification is possible if the precise target metric is known.
    \item Robust classification is impossible if the target metric is adversarially chosen from a class of metrics after the model has been trained.
\end{enumerate}

We will phrase our results in the PAC Learning Model~\cite{Valiant84}, where a learning algorithm is given labeled samples and produces a model $\model$, and later a classification algorithm $\classifier_\model$ is tasked with determining the class of a given sample $\point$. 
Robust classification for PAC learning is defined analogously as in the introduction.

Assume that a learning task admits property 1 above, that is if the learning algorithm is provided with the target metric then it can train a robust classifier in this metric. Now consider a trivial learning algorithm $\learner_0$ which stores the entire training data in the model $\model_0$ and essentially defers the learning phase to the classifier $\classifier_{\model_0}$. Once the learning algorithm is provided an adversarial example $\tilde{\point}$ and, for the sake of this outline, the target metric, the classifier $\classifier_{\model_0}$ can train a model \emph{on the fly} by running the robust learning algorithm for the target metric on the training data provided in $\model_0$. Consequently, this classifier $\classifier_{\model_0}$ will be robust in every metric in the class contradicting point 2. Thus, we need to pose a non-triviality condition on the size of the model $\model$, restricting it to be significantly smaller in size than the training data provided to the learner $\learner$.
Such a size restriction is far from exotic, as obtaining small models has always been desirable in the ML community~\cite{crowley2018pruning,frankle2018lottery,han2015learning,luo2017thinet}. 
To summarize, the question we consider is only meaningful if the size of the model $\model$ is suitably bounded. Having laid out these boundary conditions, we can now describe our results.

\begin{theorem}[Informal]
Let $\secp$ be a security parameter and let $n,\ell$ be integers, possibly depending on $\secp$, where $\ell \gg n,\secp$. Under mild cryptographic assumptions, there exists a binary learning problem $\prob$ and a class $\class$ of metrics for which
\begin{enumerate}
    \item Samples are of (small) size $n \cdot \secp$
    \item Robust classification of $\prob$ is possible with (small) models of size $\ell$ if the precise target metric is known to the learning algorithm $\learner$.
    \item Robust classification is possible with (large) models of size $n \cdot \ell$ for any metric in $\class$.
    \item There exists an \emph{efficient} adversary $\A$, which fools every efficient classifier with models of size $ < n/2 \cdot \ell$ for an adversarially chosen target metric in $\class$.
\end{enumerate}
\end{theorem}

It is instructive to think of the size parameter $\ell$ as significantly larger than all other parameters. In concrete terms, an exemplifying parameter choice is $n \approx 10^4$, $\secp = 10^3$, $k = n/4$ and $\ell = 10^{10}$. For this parameter set, the samples are of reasonable size $10^7$, a classifier with robustness against a single target metric is of size $10^{10}$, but no classifier of size smaller than $2.5 \cdot 10^{12}$ is robust against an adversarially chosen metric.

\paragraph{Perspective}

Our results show that there are learning tasks for which the only viable strategy for robust classification is to essentially include the entire training data in the model if the target metric is not precisely known. Such a classifier however fails at the essential task of condensing the information provided in the training dataset, and fails to generalize beyond the training data. Thus, the goals of compactness and robustness in an uncertain metric are fundamentally at odds.

On a technical level, we demonstrate a novel way of leveraging techniques that originate from bounded-storage model cryptography~\cite{Maurer92,CachinM97,CachinCM98,damgaard2007tight,dziembowski2002tight} to establish lower bounds on the size of the model in robust classification. Our approach deviates from prior results in this line of research and we expect it to be applicable in other settings.

\section{Technical Outline}
We will now provide an overview of our construction. The full construction with all proofs is provided in the Appendix.

\paragraph{The PAC Model} We briefly recall the Probably Approximately Correct (PAC) Learning framework~\cite{Valiant84}. A learning problem $\prob$ consists of an instance space $\mathcal{X}$, a set of classes $\mathcal{C}$ and a set of distributions $\{ \chi_c \}_{c \in \mathcal{C}}$, where each $\chi_c$ is supported on $\mathcal{X}$. We say that for a class $c \in \mathcal{C}$ $\chi_c$ samples instances of class $c$. For concreteness, we will only consider learning problems with two classes, i.e. $\mathcal{C} = \bin$. We allow the problem $\prob$ to be parametrized by a \emph{secret state} $\st$, such that one can efficiently sample from the distributions $\chi_b$ given the state $\st$. We say that $\prob$ is \emph{learnable} if there exist PPT algorithms $\learner$, called the learner, and $\classifier$, called the classifier such that the following holds. The learner $\learner$ is given labeled samples $(\point_i,b_i)$ of $\prob$, where $\point_i \pick \chi_{b_i}$, and produces a small \emph{model} $\model$. Then, the classifier $\classifier_\model$, parametrized by $\model$, is challenged with determining the class $b$ of a given sample $\point$. We define the \emph{advantage} of $\classifier_\model$ by $\Adv_\prob(\classifier_\model) = \Pr[\classifier_\model(\point) = b] - \frac{1}{2}$, where the probability is taken over the random choice of both $b \pick \bin$ and $\point \pick \chi_b$. We say that $(\learner,\classifier)$ $(\epsilon,\delta)$-PAC learns a problem $\prob$, if $\Pr[\Adv_\prob(\classifier_\model) > \epsilon] > 1 - \delta$, where the probability is taken over the choice of the training data and the random coins of $\learner$. We say that a perturbation adversary $\A$ fools a classifier $\classifier_\model$ with advantage $\epsilon$, if it hold that $\Pr[\classifier_\model(\A(\point)) = 1 - b] \geq \frac{1}{2} + \epsilon$, where the probability is taken over the random choice of $b \pick \bin$ and $\point \pick \chi_b$.

\paragraph{Weighted Hamming Metrics}
As described above, a crucial aspect of our work is that the metric which constrains the adversary is not fully specified at the time the model is trained. We will consider a simple but quite expressive class of metrics we call \emph{weighted hamming metrics}. We visualize this idea in Figure~\ref{fig:metricExample}, and continue with the formalization. Fix a finite alphabet $\Sigma$ and an integer $\feaNum$. For a vector $\point \in \Sigma^{\feaNum}$, we will call the components $\point_i$ of $\point$ the \emph{features} of $\point$. The \emph{Hamming metric} $\metric$ is defined by $\metric(\point,\OtherPoint) = \sum_{i = 1}^{\feaNum} 1_{\point_i \neq \OtherPoint_i}$ for all $\point,\OtherPoint \in \Sigma^{\feaNum}$. Here, $1_{\point_i \neq \OtherPoint_i}$ is an indicator function which assumes the value $1$ if $\point_i \neq \OtherPoint_i$ and $0$ if $\point_i = \OtherPoint_i$. We will augment the notion of Hamming metrics to \emph{weighted} Hamming metrics by introducing weights to the features. Let $\weight \in \mathbb{R}_{> 0}^{\feaNum}$ be a positive real vector. We define the weighted Hamming metric $\metric_\weight$ by
\[
\metric_\weight(\point,\OtherPoint) = \sum_{i = 1}^{\feaNum} \weightElem_i \cdot 1_{x_i \neq z_i}
\]
for all $\point, \OtherPoint \in \Sigma^{\feaNum}$. First note that $\metric_\weight$ is in fact a metric, i.e. if $\metric_{\weight}(\point,\OtherPoint) = 0$ then $\point = \OtherPoint$ and $\metric_{\weight}(\point,\mathbf{\YetAnotherPoint}) \leq \metric_{\weight}(\point,\OtherPoint) + \metric_{\weight}(\OtherPoint,\YetAnotherPoint)$ for all $\point,\OtherPoint,\YetAnotherPoint \in \Sigma^{\feaNum}$. 

Weighted Hamming metrics allow us to weigh features differently, that is, perturbing features with a high weight will be more costly for the adversary than perturbing features with small weights. To simplify matters, we will normalize the adversary's attack budget to $1$, i.e. a perturbation $\tilde{\point}$ of a sample $\point$ is permitted by metric $\metric_{\weight}$ if $\metric_{\weight}(\point,\tilde{\point}) < 1$. 

 \begin{figure}[t]
     \centering
     \input{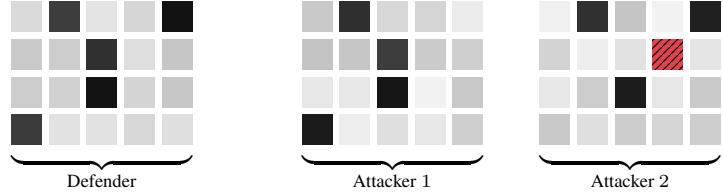}
     \caption{A simplified example of our construction.
 The data is drawn from $\mathbb{R}^{4 \times 5}$, dark squares are features where $\weight_i$ is large or $i \in \targetMSet$.
 These are defended features for the defender. \emph{Attacker 1}'s feature were known, she cannot mislead the defender's classifier.
 \emph{Attacker 2} operates under a different metric, and is able to alter a feature that the classifier is not robust in (depicted in red/hatched). Consequently, the defender is vulnerable in attacker 2's metric.}
     \label{fig:metricExample}
 \end{figure}

While weighted Hamming metrics are quite expressive in their ability to assign weights to features, we will use them in a simplified manner. Specifically, we will assign the weight $\weight_i = 1$ to a feature to \emph{protect} it, i.e. such a feature cannot be modified by the adversary. All remaining weights $\weight_i$ will be chosen as suitably small reals. More specifically, fix an integers $\feaNum$ and $\paramk$. For a subset $\targetMSet \subseteq [\feaNum]$ of size $\paramk$, define the weight $\weight_{\targetMSet}$ by 
\[
\weightElem_i = \begin{cases} 1 & \text{ if } i \in \targetMSet \\ 1/\feaNum & \text{ otherwise} \end{cases}.
\]

For simplicity, write $\metric_\targetMSet = \metric_{\weight_{\targetMSet}}$. It follows immediately from the definition of $\metric_\targetMSet$ that if $\point$ and $\tilde{\point}$ differ in a feature $i \in \targetMSet$, that is if $\point_i \neq \tilde{\point}_i$, then $\metric_\targetMSet(\point,\tilde{\point}) \geq 1$. As a consequence, any perturbation of a sample $\point$ on a feature with index $i \in \targetMSet$ will exceed the adversaries budget. Consequently, we can consider the features with indices in $\targetMSet$ as \emph{protected}. The class of metrics we consider in our constructions will be
\[
\class = \{ \metric_\targetMSet | \targetMSet \subseteq [\feaNum], |\targetMSet| = t \},
\]
i.e. every set $\targetMSet \subseteq [\feaNum]$ of size $t$ will give rise to a metric. 

\paragraph{The Basic Construction}

We will now describe a simplified version of our learning problem $\prob$. The problem has two classes, labeled 0 and 1. Let $b \in \{0,1\}$. For this construction we will use a private key encryption scheme $(\KeyGen,\Enc,\Dec)$. A sample of class $b$ consists of three encryptions $(\Enc(\key_1,b),\Enc(\key_2,b), \Enc(\key_3,b))$ of the bit $b$ under three different keys $\key_1,\key_2,\key_3$. The secret state $\st$ of the problem $\prob$ consists of the three \emph{feature keys} $\key_1,\key_2,\key_3$. In this bare-bones version, this classification problem $\prob$ is obviously unlearnable, given that the encryption scheme is secure. Thus, we will first consider a simplified setting in which the learning algorithm $\learner$ is given access to the secret state $\st = (\key_1,\key_2,\key_3)$ instead of samples from $\prob$. 

We will briefly argue that this simplification does not weaken the model. 
There is a simple transformation which augments the samples by a small amount of extra information and make the problem $\prob$ PAC learnable. The idea is to add small \emph{shares} of the secret state $\st = (\key_1,\key_2,\key_3)$ into every sample. To achieve this we will use a variant of Shamir's secret sharing~\cite{shamir1979share}. More concretely, let $\mathbb{F}$ be a finite field of size $2^\secp$ and assume we can represent $\st$ as a vector $(s_1,\dots,s_t)$ for $t = 3 \ell / \secp$, where each $s_i \in \mathbb{F}$. Define the polynomial $f(X) = \sum_{i = 1}^t s_i X^{i-1}$. Now include into every sample $\point$ the pair $(z,f(z))$, where $z \pick \mathbb{F}$ is chosen uniformly at random. Note that since both $z$ and $f(z)$ are in $\mathbb{F}$, the pair $(z,f(z))$ can be described using $2 \secp$ bits and is therefore small. Observe that given $t$ pairs $(z_1,f(z_1)),\dots,(z_{t},f(z_{t}))$ for distinct $z_i$ we can interpolate the polynomial $f$ and recover $(s_1,\dots,s_t)$ and therefore $\st = (\key_1,\key_2,\key_3)$. Since the $z_i$ are chosen uniformly at random from $\mathbb{F}$, the $z_1,\dots,z_{t}$ will be distinct except with negligible probability. It follows that $t$ samples are sufficient to recover $\st = (\key_1,\key_2,\key_3)$.

Consequently, we will henceforth only consider the simplified setting in which the learning algorithm is provided $\st = (\key_1,\key_2,\key_3)$ as input. Given the secret state $\st$, there is a simple learning algorithm for the problem $\prob$: Just pick the first key $\key_1$ and include it in the model $\model$. Then, given a sample $\point = (\ciph_1,\ciph_2,\ciph_3)$, a classifier $\classifier_\model$ can use $\key_1$ to decrypt $\ciph_1$ and obtain the class $b$. Obviously, so far there is no mechanism in place which disincentivizes the learner to include all 3 keys $\key_1,\key_2,\key_3$ in the model $\model$. Our main idea is to make the keys $\key_1,\key_2,\key_3$ \emph{very large} in order to penalize storing storing all of them in $\model$. Assume for now that each key $\key_i$ is of size $\ell \gg \feaNum$, that is storing even a single key is costly.

Now we will turn to robust classification of this problem. As the class of metrics we will consider the class $\class$ defined above, that is the metrics are of the form $\metric_\targetMSet$ for a set $\targetMSet \subseteq [3]$ of size 2. Since by construction of $\metric_\targetMSet$ the features with index $i \in \targetMSet$ are protected, this leaves just a single feature with index in the singleton set $[3] \backslash \targetMSet$ which the adversary is allowed to perturb. Note that a classifier $\classifier_\model$ in possession of all 3 keys $\key_1,\key_2,\key_3$ will be able to robustly classify this problem for any metric in the class as follows. Since every metric in the class constrains the adversary to perturbing just a single feature, we have the guarantee that 2 out of the 3 features are unmodified. Thus, given an instance $\point = (\ciph_1,\ciph_2,\ciph_3)$ $\classifier_\model$ decrypts $\ciph_1,\ciph_2$ and $\ciph_3$ obtaining bits $b_1,b_2,b_3$, and sets the bit $b$ to the majority of $b_1,b_2,b_3$. Since at most one of the $b_i$ is perturbed, $\classifier_\model$ will classify correctly.

Next assume that the model $\model$ is just big enough to store a single key $\key_i$ for $i \in [3]$. If the learner $\learner$ knows the target metric $\metric_\targetMSet$ in advance, it can choose $i$ such that $i \in \targetMSet$, i.e. the feature with index $i$ is protected and thus the key $\key_i$ is sufficient to classify robustly in the metric $\metric_\targetMSet$. But now assume that the learner does not know the target metric and the model can only store a single key. Thus, the learning algorithm needs to \emph{commit} which feature $i^\ast$ the classifier will inspect by selecting a key $\key_{i^\ast}$, and this decision cannot be altered after the fact. Thus, if the target metric $\metric_\targetMSet$ is such that $i^\ast \notin \targetMSet$, then we can construct a perturbation adversary $\A$ which fools \emph{any} classifier into misclassifying perturbed examples as follows. Given a sample $\point = (\ciph_1,\ciph_2,\ciph_3)$ for class $b$, $\A$ outputs a sample $\tilde{point} = (\ciph'_1,\ciph'_2,\ciph'_3)$ where $\ciph'_{i^\ast} = \Enc(\key_{i^\ast},1 - b)$ but $\ciph'_j = \ciph_j$ for $j \neq i^\ast$. Observe that $\metric_\targetMSet(\point,\tilde{\point}) = 1/3$, thus this modification is within the adversary's budget. Now, given the perturbed example $\tilde{\point}$, since the classifier only knows the key $\key_{i^\ast}$, it must base its decision solely on the feature $\ciph_{i^\ast}$ as by the security of the encryption scheme the contents of the other two ciphertexts are hidden from the classifier's view. Thus, from the classifier's view, $\tilde{\point}$ looks like a legit sample of class $1-b$, and it will consequently misclassify $\tilde{\point}$ as class $1-b$. In this simplified example there are only 3 possible choices for the index $i^\ast$, so the index $i^\ast$ can be found by brute force search and testing for which $i^\ast$ the classifier misclassifies. Hence, if the adversary is allowed to choose the target metric adaptively depending on the classifier $\classifier_\model$, there is an attack which perfectly fools the classifier, under the condition that the classifier only knows a single key $\key_{i^\ast}$. In this (over-)simplified analysis, we conclude that for an after-the-fact chosen target metric no classifier that takes a model $\model$ of size $\ell$ is robust, whereas we have seen above that if the target metric is known to the learner size $\ell$ suffices.

The high-level approach we have outlined here critically relies on the fact that the learning algorithm $\learner$ can only provide a bounded amount of information to the classifier $\classifier$ via the model $\model$. Leveraging memory limitations to establish security properties is a well-established research topic in cryptography. The \emph{bounded storage model}~\cite{Maurer92} admits \emph{unconditionally secure} protocols for tasks such as key-exchange~\cite{Maurer92,CachinM97,dziembowski2002tight} and secure two-party computation~\cite{CachinCM98,damgaard2007tight}, which are known to require computational assumptions in the standard model. In the same spirit as prior works that establish lower bounds~\cite{bubeck2018adversarial,DBLP:conf/colt/DegwekarNV19,mahloujifar2018can}, we make use of cryptography against the learning algorithm and classifier. But while prior works made use of cryptographic constructions that are secure against all efficient algorithms, our goal is to only establish hardness results when the model $\model$ is of bounded size.

\paragraph{The Full-Fledged Construction}

In general, we want to achieve a larger gap between the two cases. We will achieve this by modifying the problem in the following way. Instead of having just 3 features, instances of the new problem $\prob$ will have $\feaNum$ features. Furthermore, for technical reasons we will rely on a public key encryption scheme $\KeyGen,\Enc,\Dec$ rather than a private key encryption scheme. Specifically, to make our argument work we need that the adversary $\A$ can generate samples of the problem $\prob$ \emph{without} knowledge of the entire large secret state $\st = (\sk_1,\dots,\sk_\feaNum)$ but only knows small secret keys $\pk_1,\dots,\pk_\feaNum$. Thus, we will require a public key encryption scheme $(\KeyGen,\Enc,\Dec)$ which has small public keys $\pk$, but arbitrarily larger secret keys $\sk$.

For our fully-fledged construction, an instance $\point$ of the problem $\prob$ is of the form $\point = (\ciph_1,\dots,\ciph_\feaNum)$, where each $\ciph_i \gets \Enc(\pk_i,b)$ is an encryption of the class $b$ under a public key $\pk_i$. We choose the class $\class$ to consist of all metrics $\metric_\targetMSet$, where $\targetMSet \subseteq [\feaNum]$ is a set of size $\feaNum/2 + 1$. This choice of $\class$ ensures that there always is a robust classifier with model $\model$ of size $\feaNum \cdot \ell$, namely the classifier in possession of all keys $\sk_1,\dots,\sk_\feaNum$ which decrypts all $\ciph_i$ and makes a majority decision.

We now want to argue that whenever the model $\model$ is of size at most $\feaNum/2 - 1$, then for every classifier $\classifier$ there exists an \emph{efficient} adversary $\A$ which chooses a target metric $\metric_\targetMSet$ and fools $\classifier_\model$ in this metric. Recall that we allow the adversary $\A$ to make oracle-queries to $\classifier_\model$. Thus, we essentially need to construct an adversary which learns which keys $\classifier_\model$ knows by only making oracle access to $\classifier_\model$. To actually realize this idea, we will need an encryption scheme with stronger guarantees, as we will discuss in the next paragraph.

\paragraph{Big-Key Encryption}

The argument sketched so far is overly simplistic in several aspects. We assumed that the only strategy of the learner and classifier is to store the secret keys $\sk_i$ in full. In general, the standard security notion for public key encryption, indistinguishability under chosen plaintext attacks (IND-CPA security) does not provide any guarantees if the adversary is given even a small fraction of the secret key $\sk$. Recall that the learning algorithm $\learner$ is given all secret keys $\sk_1,\dots,\sk_\feaNum$ as input and thus the model $\model$ may contain a small amount of information about each of these secret keys. Thus, constructing an encryption scheme which merely has large secret keys is insufficient to force a learning algorithm to dedicate a large amount of the model $\model$ to store such keys. In particular, requiring that the secret keys are large does not preclude that there is an alternative decryption procedure which requires a significantly smaller amount of information about the secret key. Consequently, we need a stronger security property which captures the requirement that keys are large and incompressible and that missing even a small fraction of the key will render a partial key useless for decryption. 

The notion of big-key encryption~\cite{BellareKR16} offers a strong form of leakage resilience. In such a scheme, the key-generation algorithm $\KeyGen$ takes as additional input a size parameter $\ell$ and produces \emph{uniformly random} keys of size $\ell$. We note that whereas~\cite{BellareKR16} defined big-key encryption in the private key setting, we will use an analog notion for public key encryption, where we require that secret keys are very large, but both the size public keys and ciphertexts are small, only depending on the security parameter. Big-key encryption was conceived to provide strong leakage resilience guarantees and to prevent \emph{key-exfiltration attacks}. In \cite{BellareKR16} security of big-key encryption is defined via a notion called \emph{subkey prediction security}. We will use a conceptually somewhat simpler notion we call \emph{key-knowledge security}. We will briefly outline this security notion.

Recall that the standard security notion of public key encryption, indistinguishability under chosen plaintext attacks (IND-CPA security) requires that encryptions of 0 and 1 are indistinguishable for PPT distinguishers, given only the public key. An important aspect about IND-CPA security is that that the distinguisher gets no information about the secret key.

\def \leaker{\mathcal{L}}
\def \distinguisher{\mathcal{D}}
\def \extractor{\mathcal{E}}
\def \hint{\mathsf{h}}

We will define key-knowledge security via the following two stage experiment between a challenger and a pair $(\leaker,\distinguisher)$ of leaker $\leaker$ and distinguisher $\distinguisher$. The challenger generates keys $(\pk,\sk) \gets \KeyGen(1^\secp,\ell)$ and runs the leaker $\leaker$ on input $(\pk,\sk)$, an $\leaker$ will output a hint $\hint$. The distinguisher $\distinguisher_\hint$ (parametrized by $\hint$) is given the public key $\pk$ and an encryption of a random bit $b$ as input and outputs a bit $b'$. We define the advantage of $\distinguisher_\hint$ as $\Adv(\distinguisher_\hint) = |\Pr[b' = b] - 1/2|$.

We say that a big-key encryption scheme $(\KeyGen,\Enc,\Dec)$ is \emph{key-knowledge} secure, if for every pair of PPT algorithms $(\leaker,\distinguisher)$ there exists a PPT algorithm $\extractor$, called the extractor, such that the following holds. For every inverse-polynomial $\epsilon = \epsilon(\secp)$, we require that if $\Adv(\distinguisher_\hint) > \epsilon$, then $\extractor(\hint,\epsilon) = \sk$, except with negligible probability over the choice of $(\pk,\sk) \gets \KeyGen(1^\secp)$ and $\hint \gets \leaker(\pk,\sk)$. That is
\[
\Pr_{\key,\hint}[\Adv(\distinguisher_\hint) > \epsilon \text{ and } \extractor(\hint,\epsilon) \neq \sk] < \negl
\]
We allow the runtime of $\extractor$ to be $\poly[\secp,1/\epsilon]$. In essence, this security notion requires that if for a given hint $\hint$ the distinguisher $\distinguisher_\hint$ is able to distiguish encryptions of 0 and 1, then the hint $\hint$ must somehow encode the secret key $\sk$. Note in particular that this notion does not impose a size restriction on the hint $\hint$.

In the Appendix of this work we provide a construction of public key \emph{key-knowledge} secure big-key encryption in the standard model. Our construction is based on a recent construction of maliciously secure laconic conditional disclosure of secrets (laconic CDS)~\cite{DottlingGGM19}, which builds heavily on the distinguisher-dependent simulation technique~\cite{DworkNRS03,KKR17,DottlingGHMW20}. In particular, we can provide constructions of key-knowledge secure big-key encryption under standard assumptions such as the Decisional Diffie Hellman assumption~\cite{Gamal84} or the Learning with Errors assumption~\cite{Regev05}. We will omit the details of the constructions for this overview.

\paragraph{Learning the Classifier's keys}

Equipped with the notion of key-knowledge secure big-key encryption, we will complete this outline by showing that if the target metric is unknown to the learner for problem $\prob$, then there exists an \emph{efficient attack} $\A$ which fools every classifier $\classifier_\model$ with model of size at most $n/2 \cdot \ell$ in an adaptively chosen metric in the class $\class$. Recall that we allow the adversary $\A$ to make oracle queries to the classifier $\classifier_\model$. The underlying idea our attack is based on is that $\A$ can \emph{detect} which keys $\sk_i$ the classifier $\classifier_\model$ knows by making oracle access to $\classifier_\model$. This high-level idea is implemented as follows. For concreteness, we will first discuss how $\A$ can detect whether the classifier knows the key $\sk_1$. Assume that $\Adv_\prob(\classifier_\model) = \epsilon$ for some $\epsilon > 0$ and set $\gamma \ll \epsilon$, where we will determine the exact choice of $\gamma$ later.

Consider a modified problem $\prob'$, which slightly differs from $\prob$ in the way instances are sampled. To sample an instance $\point = (\ciph_1,\dots,\ciph_\feaNum)$ of class $b \in \bin$ for $\prob'$, compute $\ciph_1 \gets \Enc(\pk_1,1 - b)$ and $\ciph_i \gets \Enc(\pk_i,b)$ for all indices $i \neq 1$. That is $\ciph_1$ encrypts the flipped bit $1 - b$, whereas all other ciphertexts $\ciph_i$ encrypt the bit $b$. Thus, instances of $\prob$ and $\prob'$ differ in the first feature $\ciph_1$.

Note that given the (short) keys $\pk_1,\dots,\pk_\feaNum$, the adversary $\A$ can efficiently sample from both $\prob$ and $\prob'$. Consequently, by running $\classifier_\model$ on many samples of $\prob$ and $\prob'$, $\A$ can compute an approximation $\tilde{\epsilon}$ of $\epsilon = \Adv_{\prob}(\classifier_\model)$ and an approximation $\tilde{\epsilon}'$ of $\epsilon' = \Adv_{\prob'}(\classifier_\model)$. Using the Hoeffding bound we can establish that if the approximations were computed using $O(1/\gamma^2)$ samples, then the approximation errors are smaller than $\gamma$, except with negligible probability. That is, we can make the approximation error $\gamma$ arbitrarily small at the cost of a runtime overhead of $O(1/\gamma^2)$. Now distinguish the following two cases:
\begin{enumerate}
    \item $|\tilde{\epsilon} - \tilde{\epsilon}'| < 3 \gamma$
    \item $|\tilde{\epsilon} - \tilde{\epsilon}'| \geq 3\gamma$
\end{enumerate}
In the first case it follows that $|\epsilon - \epsilon'| < 5 \gamma$. Given that $\gamma$ is sufficiently smaller than $\epsilon$, we can treat $\epsilon'$ and $\epsilon$ as essentially the same and determine that the classifier $\classifier_\model$ is \emph{insensitive} to the modification of $\ciph_1$.
In the second case however, $\A$ will determine that $\classifier_\model$ is sensitive to the modification of $\ciph_1$ and conclude that $\classifier_\model$ must know $\sk_1$. This can be established as follows: First, note that $|\tilde{\epsilon} - \tilde{\epsilon}'| \geq 3\gamma$ implies  $|\epsilon - \epsilon'| > \gamma$. Furthermore, noting that instances of $\prob$ and $\prob'$ only differ in the first feature, we can use $\classifier_\model$ to construct a distinguisher $\distinguisher_\model$ which distinguishes $\Enc(\pk_1,0)$ and $\Enc(\pk_1,1)$ with advantage $\gamma$. However, by the key-knowledge security of the encryption scheme $(\KeyGen,\Enc,\Dec)$, there exists an extractor $\extractor$ such that $\extractor(\st,\gamma) = \sk_1$, except with negligible probability over the choice of $\sk_1$ and $\st$.

Depending on whether $\A$ determined that $\classifier_\model$ is sensitive to the modification of $\ciph_1$ or not, $\A$ will proceed as follows. If it determined that $\classifier_\model$ is \emph{insensitive} to this modification, then it will set $\prob_1 = \prob'$. On the other hand, if it determined that $\classifier_\model$ is sensitive to the modification of $\ciph_1$, it will set $\prob_1 = \prob$. Note that in either case $|\Adv_{\prob}(\classifier_\model) - \Adv_{\prob_1}(\classifier_\model)| < 3 \gamma$, so essentially $\Adv_{\prob_1}(\classifier_\model) \approx \Adv_{\prob}(\classifier_\model)$. Now $\A$ will continue this procedure for the second feature, i.e. it will modify $\prob_1$ into $\prob_1'$ by computing $\ciph_2$ via $\ciph_2 \gets \Enc(\pk_2,1-b)$. By iterating this procedure $\A$ will be able to determine all indices $i \in [\feaNum]$ for which $\classifier_\model$ is sensitive in feature $i$. Moreover, call the sequence of \emph{hybrid problems} defined in this process $\prob_1,\dots,\prob_\feaNum$.

The key insight is now the following: We claim that if the model $\model$ is of size $\leq (k+1) \cdot \ell - \secp$, then $\classifier_\model$ cannot be sensitive in more than $k$ features. To see this, note that by the above argument if $\classifier_\model$ is sensitive in $k+1$ features, say $i_1,\dots,i_{k+1}$, then we can extract the keys $\sk_{i_1},\dots,\sk_{i_{k+1}}$ from $\classifier_\model$ as described above. But now recall that the $\sk_i$ are uniformly random in $\bin^\ell$. This means that from the model $\model$, which is of size $\leq (k+1) \cdot \ell - \secp$, we can recover a uniformly random string $(\sk_{i_1},\dots,\sk_{i_{k+1}})$ which is of size $(k+1) \cdot \ell$. But by Shannon's source coding theorem~\cite{Shannon48a} this is impossible. It follows that $\classifier_\model$ is sensitive to at most $k$ features, and $\A$ can detect the indices of these features as outlined above. Let $J \subseteq [\feaNum]$ be the set of indices of these features.

Now given that $\A$ has discovered the sensitive features of $\classifier_\model$, the actual attack works as follows. Given a sample $\point = (\ciph_1,\dots,\ciph_\feaNum)$ for class $1 - b$, for $i \in J$ $\A$ replaces the $\ciph_i$ by $\ciph'_i = \Enc(\pk_i,b)$ yielding an adversarial example $\tilde{\point}$. Now note the following: By the way we constructed the set of sensitive features $J$, the adversarial example $\tilde{\point}$ is an instance of the hybrid problem $\prob_\feaNum$ for class $b$. But this means that $\Pr[\classifier_\model(\tilde{\point}) = b] \geq \frac{1}{2} + \epsilon - \feaNum \cdot \gamma$, i.e., $\tilde{\point}$ is classified as class $b$ with high probability given that $\gamma$ is suffciently smaller than $\epsilon / \feaNum$. On the other hand, let $\targetMSet \subseteq [\feaNum]$ be a set of size $\leq \feaNum - k$ such that $\targetMSet \cap J = \emptyset$, i.e., none of the indices in $J$ are protected by the metric $\metric_\targetMSet$. Then it holds that $\metric_{\targetMSet}(\point,\tilde{\point}) < 1$, i.e. in the metric $\metric_\targetMSet$ the perturbation $\tilde{\point}$ is in within $\A$'s perturbation budget.

To wrap up, we have shown that $\A$ fools any classifier $\classifier_\model$ for a model $\model$ of size at most $k \cdot \ell$ in an adaptively chosen metric $\metric_\targetMSet$ in the class $\class$. This concludes this outline.

\section{Related Work}
Little formal work has been done concerning the importance of norm choices for robustness. 
Demontis et al.~\cite{demontis2017yes} investigate the relationship between the $L_0$ and $L_{\infty}$ norms for linear classifiers.
At the same time, Croce and Hein~\cite{croce2019provable} show that a regularizer for both $L_0$ and $L_{\infty}$ can be constructed. Such a regularizer is then robust in  all $L_p$-norms.
Our work instead formalizes the problem of choosing the right metric, and the resulting problem of remaining vulnerable in another metric.

There are formal works in adversarial ML that show impossibilities to achieve robustness for specific classifiers~\cite{DBLP:conf/icml/WangJC18}.
Other works are also based on the PAC framework, but derive complexity bounds in $L_p$ ($p>1$) norms~\cite{cullina2018pac}.
Alternatively, works reason that generalization enables vulnerability~\cite{grosse2018killing} or 
that a robust classifier needs more data to train than its vulnerable counterpart~\cite{schmidt2018adversarially}.
Further, Chen et al.~\cite{chen2019towards} propose a formal argument about feature discretization, also using Hamming metrics in their reasoning.
However they derive a robustness boost given sufficiently well separated data among other properties. In this work, we show an impossibility result based on cryptographic primitives for the problem of choosing the right metric for a classifier in the context of robustness.

To conclude, we review again the works that are related to cryptographic primitives.
Most works aim to derive impossibility results for robustness.
Bubeck et al.~\cite{DBLP:conf/icml/BubeckLPR19} use the statistical query model and two statistically similar distributions.
In this setting, learning is possible, whereas robust learning is not.
This contradiction is derived by, in a nutshell, adding the label to the sample~\cite{DBLP:conf/colt/DegwekarNV19}.
Bubeck et al.~\cite{bubeck2018adversarial} rely on a similar construction, however using a pseudo-random number generator. 
The construction in our paper is instead based on ideas from cryptography in the bounded storage model~\citep{dziembowski2002tight}.
Furthermore, we extend previous settings to a learnable task that is robust to random noise.
Vulnerability then arises as the learner does not know the metric the attacker will choose in advance.
We also represent the learner as an entity that compresses the structures presented in the data, additionally to the efficiency requirement.
Another line of work rooted in cryptography aims to leverage computational hardness to increase the difficulty to compute adversarial examples.
For example, Mahloujifar and Mahmoody~\cite{mahloujifar2018can} show using signatures that computational hardness can be used to harden the task solved by an $L_0$ attacker. 
Garg et al.~\cite{garg2019adversarially} extend this work to other metrics and a game based definition for robustness.

\section{Conclusion}
In this paper, we constructed  a classification problem which admits robust classification by a small classifier if the target metric is known at training time. However, robust classification is impossible for small classifiers if the target metric is chosen after training. In the process, we explored a novel connection between hardness of robust classification an bounded storage model cryptography.


\section*{Acknowledgments}
This work was supported by the German Federal Ministry of Education and
Research (BMBF) through funding for the Center for IT-Security,
Privacy and Accountability (CISPA) (FKZ: 16KIS0753). This work is partially funded by the Helmholtz Association within the project "Trustworthy Federated Data Analytics” (TFDA) (funding number ZT-I-OO1 4).

\bibliography{lit}

\begin{thebibliography}{53}
\providecommand{\natexlab}[1]{#1}
\providecommand{\url}[1]{\texttt{#1}}
\expandafter\ifx\csname urlstyle\endcsname\relax
  \providecommand{\doi}[1]{doi: #1}\else
  \providecommand{\doi}{doi: \begingroup \urlstyle{rm}\Url}\fi

\bibitem[Athalye et~al.(2018)Athalye, Carlini, and
  Wagner]{athalye2018obfuscated}
A.~Athalye, N.~Carlini, and D.~Wagner.
\newblock Obfuscated gradients give a false sense of security: Circumventing
  defenses to adversarial examples.
\newblock In \emph{ICML}, 2018.

\bibitem[Bellare et~al.(2016)Bellare, Kane, and Rogaway]{BellareKR16}
M.~Bellare, D.~Kane, and P.~Rogaway.
\newblock Big-key symmetric encryption: Resisting key exfiltration.
\newblock In \emph{{CRYPTO} {(1)}}, volume 9814 of \emph{Lecture Notes in
  Computer Science}, pages 373--402. Springer, 2016.

\bibitem[Biggio and Roli(2018)]{biggio2018wild}
B.~Biggio and F.~Roli.
\newblock Wild patterns: Ten years after the rise of adversarial machine
  learning.
\newblock \emph{Pattern Recognition}, 84:\penalty0 317--331, 2018.

\bibitem[Bubeck et~al.(2018)Bubeck, Lee, Price, and
  Razenshteyn]{bubeck2018adversarial}
S.~Bubeck, Y.~T. Lee, E.~Price, and I.~Razenshteyn.
\newblock Adversarial examples from cryptographic pseudo-random generators.
\newblock \emph{arXiv preprint arXiv:1811.06418}, 2018.

\bibitem[Bubeck et~al.(2019)Bubeck, Lee, Price, and
  Razenshteyn]{DBLP:conf/icml/BubeckLPR19}
S.~Bubeck, Y.~T. Lee, E.~Price, and I.~P. Razenshteyn.
\newblock Adversarial examples from computational constraints.
\newblock In \emph{ICML}, pages 831--840, 2019.

\bibitem[Cachin and Maurer(1997)]{CachinM97}
C.~Cachin and U.~M. Maurer.
\newblock Unconditional security against memory-bounded adversaries.
\newblock In \emph{{CRYPTO}}, volume 1294 of \emph{Lecture Notes in Computer
  Science}, pages 292--306. Springer, 1997.

\bibitem[Cachin et~al.(1998)Cachin, Cr{\'{e}}peau, and Marcil]{CachinCM98}
C.~Cachin, C.~Cr{\'{e}}peau, and J.~Marcil.
\newblock Oblivious transfer with a memory-bounded receiver.
\newblock In \emph{{FOCS}}, pages 493--502. {IEEE} Computer Society, 1998.

\bibitem[Carlini and Wagner(2017)]{carlini2017adversarial}
N.~Carlini and D.~Wagner.
\newblock Adversarial examples are not easily detected: Bypassing ten detection
  methods.
\newblock In \emph{Proceedings of the 10th ACM Workshop on Artificial
  Intelligence and Security}, pages 3--14. ACM, 2017.

\bibitem[Chen et~al.(2019)Chen, Wu, Rastogi, Liang, and Jha]{chen2019towards}
J.~Chen, X.~Wu, V.~Rastogi, Y.~Liang, and S.~Jha.
\newblock Towards understanding limitations of pixel discretization against
  adversarial attacks.
\newblock In \emph{EuroS\&P}, pages 480--495. IEEE, 2019.

\bibitem[Cohen et~al.(2019)Cohen, Rosenfeld, and Kolter]{cohen2019certified}
J.~M. Cohen, E.~Rosenfeld, and J.~Z. Kolter.
\newblock Certified adversarial robustness via randomized smoothing.
\newblock In \emph{Proceedings of the 36th International Conference on Machine
  Learning, {ICML} 2019}, 2019.

\bibitem[Croce and Hein(2019)]{croce2019provable}
F.~Croce and M.~Hein.
\newblock Provable robustness against all adversarial {$ l_p $}-perturbations
  for {$p \geq 1$}.
\newblock \emph{arXiv preprint arXiv:1905.11213}, 2019.

\bibitem[Crowley et~al.(2018)Crowley, Turner, Storkey, and
  O'Boyle]{crowley2018pruning}
E.~J. Crowley, J.~Turner, A.~J. Storkey, and M.~F.~P. O'Boyle.
\newblock Pruning neural networks: is it time to nip it in the bud?
\newblock \emph{NIPS 2018 Workshop CDNNRIA}, 2018.

\bibitem[Cullina et~al.(2018)Cullina, Bhagoji, and Mittal]{cullina2018pac}
D.~Cullina, A.~N. Bhagoji, and P.~Mittal.
\newblock Pac-learning in the presence of adversaries.
\newblock In \emph{Advances in Neural Information Processing Systems}, pages
  230--241, 2018.

\bibitem[Damg{\aa}rd et~al.(2007)Damg{\aa}rd, Fehr, Renner, Salvail, and
  Schaffner]{damgaard2007tight}
I.~B. Damg{\aa}rd, S.~Fehr, R.~Renner, L.~Salvail, and C.~Schaffner.
\newblock A tight high-order entropic quantum uncertainty relation with
  applications.
\newblock In \emph{Annual International Cryptology Conference}, pages 360--378.
  Springer, 2007.

\bibitem[Degwekar et~al.(2019)Degwekar, Nakkiran, and
  Vaikuntanathan]{DBLP:conf/colt/DegwekarNV19}
A.~Degwekar, P.~Nakkiran, and V.~Vaikuntanathan.
\newblock Computational limitations in robust classification and win-win
  results.
\newblock In \emph{COLT}, pages 994--1028, 2019.

\bibitem[Demontis et~al.(2017)Demontis, Melis, Biggio, Maiorca, Arp, Rieck,
  Corona, Giacinto, and Roli]{demontis2017yes}
A.~Demontis, M.~Melis, B.~Biggio, D.~Maiorca, D.~Arp, K.~Rieck, I.~Corona,
  G.~Giacinto, and F.~Roli.
\newblock Yes, machine learning can be more secure! a case study on android
  malware detection.
\newblock \emph{IEEE Transactions on Dependable and Secure Computing}, 2017.

\bibitem[D{\"{o}}ttling et~al.(2019)D{\"{o}}ttling, Garg, Goyal, and
  Malavolta]{DottlingGGM19}
N.~D{\"{o}}ttling, S.~Garg, V.~Goyal, and G.~Malavolta.
\newblock Laconic conditional disclosure of secrets and applications.
\newblock In \emph{{FOCS}}, pages 661--685. {IEEE} Computer Society, 2019.

\bibitem[D{\"{o}}ttling et~al.(2020)D{\"{o}}ttling, Garg, Hajiabadi, Masny, and
  Wichs]{DottlingGHMW20}
N.~D{\"{o}}ttling, S.~Garg, M.~Hajiabadi, D.~Masny, and D.~Wichs.
\newblock Two-round oblivious transfer from {CDH} or {LPN}.
\newblock In \emph{{EUROCRYPT} {(2)}}, volume 12106 of \emph{Lecture Notes in
  Computer Science}, pages 768--797. Springer, 2020.

\bibitem[Dwork et~al.(2003)Dwork, Naor, Reingold, and Stockmeyer]{DworkNRS03}
C.~Dwork, M.~Naor, O.~Reingold, and L.~J. Stockmeyer.
\newblock Magic functions.
\newblock \emph{J. {ACM}}, 50\penalty0 (6):\penalty0 852--921, 2003.

\bibitem[Dziembowski and Maurer(2002)]{dziembowski2002tight}
S.~Dziembowski and U.~Maurer.
\newblock Tight security proofs for the bounded-storage model.
\newblock In \emph{Proceedings of the thiry-fourth annual ACM symposium on
  Theory of computing}, pages 341--350, 2002.

\bibitem[Frankle and Carbin(2019)]{frankle2018lottery}
J.~Frankle and M.~Carbin.
\newblock The lottery ticket hypothesis: Finding sparse, trainable neural
  networks.
\newblock \emph{ICLR}, 2019.

\bibitem[Gamal(1984)]{Gamal84}
T.~E. Gamal.
\newblock A public key cryptosystem and a signature scheme based on discrete
  logarithms.
\newblock In \emph{{CRYPTO}}, volume 196 of \emph{Lecture Notes in Computer
  Science}, pages 10--18. Springer, 1984.

\bibitem[Garg et~al.(2013)Garg, Gentry, Sahai, and Waters]{GGSW13}
S.~Garg, C.~Gentry, A.~Sahai, and B.~Waters.
\newblock Witness encryption and its applications.
\newblock In \emph{{STOC}}, pages 467--476. {ACM}, 2013.

\bibitem[Garg et~al.(2019)Garg, Jha, Mahloujifar, and
  Mahmoody]{garg2019adversarially}
S.~Garg, S.~Jha, S.~Mahloujifar, and M.~Mahmoody.
\newblock Adversarially robust learning could leverage computational hardness.
\newblock \emph{arXiv preprint arXiv:1905.11564}, 2019.

\bibitem[Glass and Switkes(1976)]{glass1976pattern}
L.~Glass and E.~Switkes.
\newblock Pattern recognition in humans: Correlations which cannot be
  perceived.
\newblock \emph{Perception}, 5\penalty0 (1):\penalty0 67--72, 1976.

\bibitem[Goldreich et~al.(1984)Goldreich, Goldwasser, and
  Micali]{GoldreichGM84}
O.~Goldreich, S.~Goldwasser, and S.~Micali.
\newblock On the cryptographic applications of random functions.
\newblock In \emph{{CRYPTO}}, volume 196 of \emph{Lecture Notes in Computer
  Science}, pages 276--288. Springer, 1984.

\bibitem[Goldwasser et~al.(2013)Goldwasser, Kalai, Popa, Vaikuntanathan, and
  Zeldovich]{GKPVZ13}
S.~Goldwasser, Y.~T. Kalai, R.~A. Popa, V.~Vaikuntanathan, and N.~Zeldovich.
\newblock How to run turing machines on encrypted data.
\newblock In \emph{{CRYPTO} {(2)}}, volume 8043 of \emph{Lecture Notes in
  Computer Science}, pages 536--553. Springer, 2013.

\bibitem[Grosse et~al.(2018)Grosse, Smith, and Backes]{grosse2018killing}
K.~Grosse, M.~T. Smith, and M.~Backes.
\newblock Killing four birds with one gaussian process: Analyzing test-time
  attack vectors on classification.
\newblock \emph{arXiv preprint arXiv:1806.02032}, 2018.

\bibitem[Han et~al.(2015)Han, Pool, Tran, and Dally]{han2015learning}
S.~Han, J.~Pool, J.~Tran, and W.~Dally.
\newblock Learning both weights and connections for efficient neural network.
\newblock In \emph{NIPS}, pages 1135--1143, 2015.

\bibitem[He et~al.(2017)He, Wei, Chen, Carlini, and Song]{he2017adversarial}
W.~He, J.~Wei, X.~Chen, N.~Carlini, and D.~Song.
\newblock Adversarial example defense: Ensembles of weak defenses are not
  strong.
\newblock In \emph{11th $\{$USENIX$\}$ Workshop on Offensive Technologies
  ($\{$WOOT$\}$ 17)}, 2017.

\bibitem[Jain et~al.(2017)Jain, Kalai, Khurana, and Rothblum]{KKR17}
A.~Jain, Y.~T. Kalai, D.~Khurana, and R.~Rothblum.
\newblock Distinguisher-dependent simulation in two rounds and its
  applications.
\newblock In \emph{{CRYPTO} {(2)}}, volume 10402 of \emph{Lecture Notes in
  Computer Science}, pages 158--189. Springer, 2017.

\bibitem[Jia et~al.(2020)Jia, Wang, Cao, and Gong]{jia2020certified}
J.~Jia, B.~Wang, X.~Cao, and N.~Z. Gong.
\newblock Certified robustness of community detection against adversarial
  structural perturbation via randomized smoothing.
\newblock In \emph{Proceedings of The Web Conference 2020}, pages 2718--2724,
  2020.

\bibitem[Lee et~al.(2019)Lee, Yuan, Chang, and Jaakkola]{lee2019tight}
G.-H. Lee, Y.~Yuan, S.~Chang, and T.~Jaakkola.
\newblock Tight certificates of adversarial robustness for randomly smoothed
  classifiers.
\newblock In \emph{Advances in Neural Information Processing Systems}, pages
  4911--4922, 2019.

\bibitem[Levine and Feizi(2020)]{levine2020wasserstein}
A.~Levine and S.~Feizi.
\newblock Wasserstein smoothing: Certified robustness against wasserstein
  adversarial attacks.
\newblock In \emph{International Conference on Artificial Intelligence and
  Statistics}, pages 3938--3947, 2020.

\bibitem[Luo et~al.(2017)Luo, Wu, and Lin]{luo2017thinet}
J.-H. Luo, J.~Wu, and W.~Lin.
\newblock Thinet: A filter level pruning method for deep neural network
  compression.
\newblock In \emph{ICCV}, pages 5058--5066, 2017.

\bibitem[Mahloujifar and Mahmoody(2019)]{mahloujifar2018can}
S.~Mahloujifar and M.~Mahmoody.
\newblock Can adversarially robust learning leveragecomputational hardness?
\newblock In \emph{Algorithmic Learning Theory}, pages 581--609, 2019.

\bibitem[Maurer(1992)]{Maurer92}
U.~M. Maurer.
\newblock Conditionally-perfect secrecy and a provably-secure randomized
  cipher.
\newblock \emph{J. Cryptology}, 5\penalty0 (1):\penalty0 53--66, 1992.

\bibitem[{Papernot} et~al.(2018){Papernot}, {McDaniel}, {Sinha}, and
  {Wellman}]{8406613}
N.~{Papernot}, P.~{McDaniel}, A.~{Sinha}, and M.~P. {Wellman}.
\newblock Sok: Security and privacy in machine learning.
\newblock In \emph{2018 IEEE European Symposium on Security and Privacy (EuroS
  P)}, pages 399--414, April 2018.

\bibitem[Regev(2005)]{Regev05}
O.~Regev.
\newblock On lattices, learning with errors, random linear codes, and
  cryptography.
\newblock In \emph{{STOC}}, pages 84--93. {ACM}, 2005.

\bibitem[Salman et~al.(2019)Salman, Li, Razenshteyn, Zhang, Zhang, Bubeck, and
  Yang]{salman2019provably}
H.~Salman, J.~Li, I.~Razenshteyn, P.~Zhang, H.~Zhang, S.~Bubeck, and G.~Yang.
\newblock Provably robust deep learning via adversarially trained smoothed
  classifiers.
\newblock In \emph{NeurIPS}, pages 11289--11300, 2019.

\bibitem[Schmidt et~al.(2018)Schmidt, Santurkar, Tsipras, Talwar, and
  Madry]{schmidt2018adversarially}
L.~Schmidt, S.~Santurkar, D.~Tsipras, K.~Talwar, and A.~Madry.
\newblock Adversarially robust generalization requires more data.
\newblock In \emph{NIPS}, pages 5014--5026, 2018.

\bibitem[Shalev-Shwartz and Ben-David(2014)]{shalev2014understanding}
S.~Shalev-Shwartz and S.~Ben-David.
\newblock \emph{Understanding machine learning: From theory to algorithms}.
\newblock Cambridge university press, 2014.

\bibitem[Shamir(1979)]{shamir1979share}
A.~Shamir.
\newblock How to share a secret.
\newblock \emph{Communications of the ACM}, 22\penalty0 (11):\penalty0
  612--613, 1979.

\bibitem[Shannon(1948)]{Shannon48a}
C.~E. Shannon.
\newblock A mathematical theory of communication.
\newblock \emph{Bell Syst. Tech. J.}, 27\penalty0 (4):\penalty0 623--656, 1948.

\bibitem[Sharma and Chen(2018)]{sharma2017attacking}
Y.~Sharma and P.~Chen.
\newblock Attacking the madry defense model with $l_1$-based adversarial
  examples.
\newblock In \emph{ICLR, 2018, Workshop Track Proceedings}, 2018.

\bibitem[Tramer et~al.(2020)Tramer, Carlini, Brendel, and
  Madry]{tramer2020adaptive}
F.~Tramer, N.~Carlini, W.~Brendel, and A.~Madry.
\newblock On adaptive attacks to adversarial example defenses.
\newblock \emph{arXiv preprint arXiv:2002.08347}, 2020.

\bibitem[Valiant(1984)]{Valiant84}
L.~G. Valiant.
\newblock A theory of the learnable.
\newblock In \emph{{STOC}}, pages 436--445. {ACM}, 1984.

\bibitem[{Wang} et~al.(2019){Wang}, {Yi}, {Zou}, and {Wu}]{8845708}
K.~{Wang}, P.~{Yi}, F.~{Zou}, and Y.~{Wu}.
\newblock Generating adversarial samples with constrained wasserstein distance.
\newblock \emph{IEEE Access}, 7:\penalty0 136812--136821, 2019.

\bibitem[Wang et~al.(2018)Wang, Jha, and Chaudhuri]{DBLP:conf/icml/WangJC18}
Y.~Wang, S.~Jha, and K.~Chaudhuri.
\newblock Analyzing the robustness of nearest neighbors to adversarial
  examples.
\newblock In \emph{ICML}, pages 5120--5129, 2018.

\bibitem[Wang et~al.(2004)Wang, Bovik, Sheikh, and
  Simoncelli]{DBLP:journals/tip/WangBSS04}
Z.~Wang, A.~C. Bovik, H.~R. Sheikh, and E.~P. Simoncelli.
\newblock Image quality assessment: from error visibility to structural
  similarity.
\newblock \emph{{IEEE} Trans. Image Processing}, 13\penalty0 (4):\penalty0
  600--612, 2004.

\bibitem[Wong et~al.(2019)Wong, Schmidt, and Kolter]{wong2019wasserstein}
E.~Wong, F.~R. Schmidt, and J.~Z. Kolter.
\newblock Wasserstein adversarial examples via projected sinkhorn iterations.
\newblock \emph{arXiv preprint arXiv:1902.07906}, 2019.

\bibitem[Xiao et~al.(2018)Xiao, Zhu, Li, He, Liu, and Song]{Xiao:2018up}
C.~Xiao, J.-Y. Zhu, B.~Li, W.~He, M.~Liu, and D.~Song.
\newblock {Spatially Transformed Adversarial Examples}.
\newblock \emph{arXiv.org}, Jan. 2018.

\bibitem[Zhang et~al.(2018)Zhang, Isola, Efros, Shechtman, and
  Wang]{zhang2018unreasonable}
R.~Zhang, P.~Isola, A.~A. Efros, E.~Shechtman, and O.~Wang.
\newblock The unreasonable effectiveness of deep features as a perceptual
  metric.
\newblock In \emph{CVPR}, pages 586--595, 2018.

\end{thebibliography}

\clearpage
\appendix
\section{Overview}
This Appendix contains the formal proofs and constructions with their corresponding background. We split this content into three parts, of which we here give a brief overview.

\textbf{B - Encryption with Key-Knowledge Security.} Our construction is based on a public key encryption scheme. 
To encompass the size constraint in relation to robustness, we require  
secret keys to be arbitrarily large, whereas the ciphertexts and public keys are compact.
This leaves us with the problem that a partial key might suffice to reconstruct the original message or sample.
We thus require additionally a strengthened security guarantee that any adversary that is able to distinguish encrypted messages with non-negligible advantage must know the whole corresponding secret key. 
The first section formalizes and verifies these properties to lay the foundation for our constructions.

\textbf{C - Definition of PAC and robust PAC learning.} In addition to the previously formalized encryption scheme, we also define learning and in particular robust learning. To this end, we use the Probably Approximately Correct (PAC) learning framework~\cite{Valiant84}. We distinguish two notions of robustness, strongly robust PAC learning and weakly robust PAC learning. Roughly speaking, a strongly robust classifier cannot be fooled at all. A weakly robust classifier represents a more realistic scenario, where the classifier contains some form of outlier detection. In the case of weakly robust PAC learning, the attacker is able to succeed if she reliably triggers this outlier class (comparable to DDoS attack on a server), or alternatively is able to craft a confidently classified example.

\textbf{D - Our Construction.} This last section combines the building blocks into our construction. We first define a learning task and the small classifier that is robust in one metric. 
Then, we describe the large classifier that is robust in any metric and conclude with the impossibility result.

Before we start with the encryption scheme, we recap the definition of weighted hamming metrics as given in the main paper. 
The metric is defined over a finite alphabet $\Sigma$. 
We write vector $\point \in \Sigma^{\feaNum}$, with $\feaNum$ components or features, each referred as $\point_i$.
For two $\point,\OtherPoint \in \Sigma^{\feaNum}$, we define
the \emph{Hamming metric} $\metric$ as $\metric(\point,\OtherPoint) = \sum_{i = 1}^{\feaNum} 1_{\point_i \neq \OtherPoint_i}$. 
The metric is based on the indicator function $1_{\point_i \neq \OtherPoint_i}$ which assumes the value $1$ if $\point_i \neq \OtherPoint_i$ and $0$ otherwise.
As stated in the main paper, we augment the notion of Hamming metrics to \emph{weighted} Hamming metrics. Let $\weight \in \mathbb{R}_{> 0}^{\feaNum}$ be a positive real vector. We define the weighted Hamming metric $\metric_\weight$ as
\[
\metric_\weight(\point,\OtherPoint) = \sum_{i = 1}^{\feaNum} \weightElem_i \cdot 1_{x_i \neq z_i}
\]
for all $\point, \OtherPoint \in \Sigma^{\feaNum}$. 
Note that $\metric_\weight$ is in fact a metric, i.e. if $\metric_{\weight}(\point,\OtherPoint) = 0$ then $\point = \OtherPoint$ and $\metric_{\weight}(\point,\mathbf{\YetAnotherPoint}) \leq \metric_{\weight}(\point,\OtherPoint) + \metric_{\weight}(\OtherPoint,\YetAnotherPoint)$ for all $\point,\OtherPoint,\YetAnotherPoint \in \Sigma^{\feaNum}$. 

The weights of the metric are directly linked to the attacker:
perturbing a feature with a high weight will be more costly for the adversary than perturbing features with small weights. 
We further normalize the adversary's attack budget to $1$.
In other words, a perturbation $\tilde{\point}$ is allowed if 
 $\metric_{\weight}(\point,\tilde{\point}) < 1$. 

We further simplify the weights and
assign weight $\weight_i = 1$ to a feature to protect it:
if $\weight_i = 1$, $x_i$ cannot be modified by the adversary.
Given integers $\feaNum$ and $\paramk$, for a subset 
 $\targetMSet \subseteq [\feaNum]$ of size $\paramk$, define the weight $\weight_{\targetMSet}$ by 
\[
\weightElem_i = \begin{cases} 1 & \text{ if } i \in \targetMSet \\ 1/\feaNum & \text{ otherwise} \end{cases}.
\]

For simplicity, we write $\metric_\targetMSet = \metric_{\weight_{\targetMSet}}$. 
For two points $\point$ and $\tilde{\point}$, $\metric_\targetMSet(\point,\tilde{\point}) \geq 1$
if both points differ in a feature  $i \in \targetMSet$, hence $\point_i \neq \tilde{\point}_i$.
In other words, the adversaries budget is exceeded if a feature with index $i \in \targetMSet$ is perturbed: the feature is protected.
 The class of metrics we consider in our constructions is
\[
\class = \{ \metric_\targetMSet | \targetMSet \subseteq [\feaNum], |\targetMSet| = t \},
\]
i.e. every set $\targetMSet \subseteq [\feaNum]$ of size $t$ will give rise to a metric.

\section{Big-Key Encryption with Key-Knowledge Security}

\def \st{\mathsf{st}}

In this Section we will discuss a type of public key encryption scheme we call \emph{encryption with key knowledge}. This is a standard encryption scheme which comes with the following strengthened security guarantee: Any adversary which distinguishes encryptions of (say) 0 and 1 with non-negligible advantage \emph{must know the corresponding secret key}. We will construct such encryption schemes with arbitrarily large secret keys but compact ciphertexts. Conforming with previous works, we will call this type of encryption \emph{big-key encryption}. For the sake of simplicity, we we only define big-key encryption for binary messages $\mes \in \bin$.

\begin{definition}
A big-key encryption scheme consists of 3 algorithms $(\KeyGen,\Enc,\Dec)$ with the following syntax.
\begin{itemize}
    \item $\KeyGen(1^\secp,\ell)$: Takes as input a security parameter $1^\secp$ and a size-parameter $\ell$ and outputs a public key $\pk$ and a secret key $\sk$.
    \item $\Enc(\pk,\mes)$: Takes as input a public key $\pk$ and a message $\mes \in \bin$ and outputs a ciphertext $\ciph$
    \item $\Dec(\sk,\ciph)$: Takes as input a secret key $\sk$ and a ciphertext $\ciph$ and outputs a message $\mes' \in \bin$.
\end{itemize}
Assume for simplicity that the message space is $\bin$. We require the following properties.
\begin{itemize}
    \item \textbf{Correctness}: It holds for every message $\mes \in \bin$ that $\Dec(\sk,\Enc(\pk,\mes)) = \mes$, where $(\pk,\sk) \gets \KeyGen(1^\secp,\ell)$.
    \item \textbf{Compactness}: Secret keys $\sk$ are of size at most $\ell \cdot \secp$. Public keys $\pk$ and ciphertexts $\ciph$ are of size $\poly$ and in particular independent of $\ell$. Furthermore, we require that the runtime of $\Enc$ is also $\poly$ independent of $\ell$.
\end{itemize}
\end{definition}

\def \Adv{\mathsf{Adv}}
\def \Ex{\mathcal{E}}

We define security of big-key encryption via a notion we call key-knowledge security. In a nutshell, this notion requires that an adversary who can distinguish encryptions of 0 and 1 under a public key $\pk$ must know the corresponding secret key $\sk$. This is formalized via a \emph{knowledge extractor} $\Ex$.

\begin{definition}[Key-Knowledge Security]
Let $(\KeyGen,\Enc,\Dec)$ be a big-key encryption scheme. Consider the following 2-phase security experiment with a two-stage adversary $\A = (\A_1,\A_2)$.
\begin{description}
\item Stage $S_1(1^\secp,\A_1)$:
\begin{itemize}
\item Compute $(\pk,\sk) \gets \KeyGen(1^\secp,\ell)$
\item $\st \gets \A_1(1^\secp,\pk,\sk)$
\item Output $(\st,\pk)$
\end{itemize} 
\item Stage $S_2(\A_2,b,\st,\pk)$:
\begin{itemize}
\item Compute $\ciph^\ast \gets \Enc(\pk,b)$
\item Compute $b' \gets \A_2(\st,\pk,\ciph^\ast)$
\item If $b' = b$ output 1, otherwise 0.
\end{itemize}
\end{description}

Assume without loss of generality that $\st$ contains $\pk$. Fix an intermediate output $(\st,\pk) \gets S_1(1^\secp,\A_1)$ and let $\Exp_\A(\st)$ be the output of the experiment with intermediate state $\st$, i.e. $\Exp^b_\A(\st) = S_2(\A_2,b,\st)$. The advantage of $\A_2$ is defined by
\[
\Adv_\st(\A_2) = \Pr[\Exp^0_\A(\st) = 1] - \Pr[\Exp^1_\A(\st) = 1].
\]

We say that a big-key encryption scheme $(\KeyGen,\Enc,\Dec)$ is \emph{key-knowledge secure}, if there exists a PPT extractor $\Ex$, such that for every PPT adversary $\A = (\A_1,\A_2)$ the following holds. For every inverse-polynomial $\epsilon$ it holds that
\[
\Pr[\Adv_\st(\A) > \epsilon \text{ and } \Ex(\A,\epsilon) \neq \sk] < \epsilon,
\]
except for finitely many $\secp$. Here the probability is taken over the random choice of $\st$. Here, the runtime of $\Ex(\A,\epsilon)$ is $\poly[\secp,1/\epsilon]$.
\end{definition}

\subsection{Malicious Laconic Conditional Disclosure of Secrets}

In order to construct key-knowledge secure big-key encryption, we will make use of a recently introduced primitive called malicious laconic conditional disclosure of secrets, or lCDS for for short~\cite{DottlingGGM19}.

An lCDS scheme lets can be seen as a two round witness-encryption scheme~\cite{GGSW13}, in which the first message of the receiver commits to the witness. The feature of interest of lCDS is that the size of both the commitment and ciphertexts is independent of the size of the witness. This almost immediately implies a big-key encryption scheme. All we need additionally is an NP-language which has small statements but large and incompressible witnesses. We can construct such a language using collision resistant hash functions.

\def \lCDS{\mathsf{lCDS}}
\def \Rec{\mathsf{Rec}}
\def \crs{\mathsf{crs}}
\def \statement{\mathsf{x}}
\def \witness{\mathsf{w}}
\def \lang{\mathcal{L}}
\def \witrel{\mathcal{R}}
\def \com{\mathsf{com}}

\begin{definition}
Let $\lang$ be an NP-language and let $\witrel_\lang$ be its witness-relation. An laconic CDS scheme $\lCDS$ consists of four algorithms $(\Setup,\Rec,\Enc,\Dec)$ with the following syntax.
\begin{description}
    \item[$\Setup(1^\secp)$]: Takes as input the security paramter $1^\secp$ and outputs a common reference string $\crs$.
    
    \item[$\Rec(\crs,\statement,\witness)$]: Takes as input a common reference string $\crs$, a statement $\statement$ and a witness $\witness$ and outputs a a commitment $\com$ and a state $\st$.
    
    \item[$\Enc(\crs,\statement,\com,\mes)$]: Takes as input a common reference string $\crs$, a statement $\statement$, a commitment $\com$ and a message $\mes$ and outputs a ciphertext $\ciph$.
    
    \item[$\Dec(\crs,\ciph,\st)$]: Takes as input a common reference string $\crs$, a ciphertext $\ciph$ and a state $\st$ and outputs a message $\mes'$.
\end{description}

We require the following properties of laconic CDS scheme.
\begin{itemize}
    \item \textbf{Correctness}: It holds that $\Pr[\Dec(\crs,\Enc(\crs,\statement,\com,\mes),\st) = \mes] = 1$ given that $\crs \gets \Setup(1^\secp)$ and $(\com,\st) \gets \Rec(\crs,\statement,\witness)$.
    
\item \textbf{Compactness}: It holds that $|\com|$ and $|\ciph|$ are of size $\poly$ and in particular \emph{independent} of $|w|$.
\end{itemize}
\end{definition}

Note that in general the state $\st$ could be substantially larger than the witness $\witness$. However, without loss of generality we can assume that the state is of size $|\statement| + |\witness| + \secp$, as the $\st$ can be recomputed from $\crs, \statement, \witness$ and the random coins of $\Rec$, which we can assume to be a PRG-seed of size $\secp$.

In \cite{DottlingGGM19}, the security of laconic CDS and laconic functionalities in general is defined via a notion called context security. We briefly recall this definition and then show how it can be simplified for our purposes.

\def \mc{\mathcal}
\def \mcZ{\mathcal{Z}}
\def \ms{\mathsf}
\def \coins{\ms{r}}
\def \crs{\ms{crs}}
\def \rec{\ms{rec}}
\def \snd{\ms{snd}}

\def \ExtZ{\ms{Ext}\mc{Z}}
\def \mcEZ{\mc{EZ}}

\def \Sim{\ms{Sim}}
\def \aux{\ms{aux}}

\begin{definition}[Protocol Context]
We say that a PPT machine $\mcZ = (\mcZ_1,\mcZ_2)$ is a context for two message protocols $\Pi = (\Setup,\ms{R}_1,\ms{S},\ms{R}_2)$, if it has the following syntactic properties: The first stage $\mcZ_1$ takes as input a common reference string $\crs$ (generated by $\Setup$) and random coins $\coins_1$ and outputs a receiver message $\rec$ and a state $\st$. The second phase $\mcZ_2$ takes as input the state $\st$ and random coins $\coins_2$. The second phase is allowed to make queries $y$ to a sender oracle $\mc{O}_{\crs,\rec}(y)$, which are answered by $\ms{S}(\crs,\rec,y)$ (using fresh randomness from $\coins_2$). In the end the context outputs a bit $b^\ast$. Define $\mcZ(\secparam)$ by
\begin{itemize}
    \item Choose random tapes $\coins_1, \coins_2$
    \item Compute $\crs \gets \Setup(\secparam)$
    \item $(\st,\rec) \gets \mcZ_1(\crs,\coins_1)$
    \item $b^\ast \gets \mcZ_2^{\mc{O}_{\crs,\rec}(\cdot)}(\st,\coins_2)$
    \item Output $b^\ast$
\end{itemize}
\end{definition}

We will now provide our definition of context security.
\begin{definition}\label{def:key_knowledge}
Let $\Pi = (\Setup,\ms{R}_1,\ms{S},\ms{R}_2)$ be a two-message protocol realizing a two-party functionality $\mc{F}$. We say that $\Pi$ is \emph{context-secure} if the following holds for every $\Pi$-context $\mc{Z} = (\mc{Z}_1,\mc{Z}_2)$. We require that there exists a context extractor $\ExtZ$ and a simulators $\Sim$ such that the following holds for every $\delta > 0$:
\begin{enumerate}[(1)]

    \item $\ExtZ$ takes as input $\crs,\st,\rec$, random coins $\coins^\ast$ and a parameter $\delta$ and outputs a value $x^\ast$ and an auxiliary string $\aux$. $\ExtZ$ has overhead $\poly \cdot T_2$, where $T_2$ is the overhead of $\mcZ_2$ and the polynomial is independent of $\mcZ$, only depending on $\secpar$. 
    \item $\Sim$ takes as input $\rec,\aux$ and a value $z$ and outputs a sender-message $\snd$. We require The overhead of $\Sim$ to be polynomial in the overhead of $\secpar$, but independent of $\mcZ_1$ and $\mcZ_2$. 
    \item The experiment $\mcEZ(\secparam,\delta)$ is defined by
    \begin{itemize}
        \item Choose random tapes $\coins_1,\coins_2$
        \item Compute $\crs \gets \Setup(\secparam)$
        \item $(\st,\rec) \gets \mcZ_1(\crs,\coins_1)$
        \item $(x^\ast,\aux) \gets \ExtZ(\crs,\st,\rec, \coins^\ast,\delta)$
        \item $b^\ast \gets \mcZ_2^{\mc{O}'(\cdot)}(\st,\coins_2)$, where $\mc{O}'(y)$ computes and outputs $\Sim(\rec,\aux,\mc{F}(x^\ast,y))$
        \item Output $b^\ast$
    \end{itemize}
    \item (Security) It holds for every inverse polynomial $\epsilon = \epsilon(\secpar)$ that
    \[
    |\Pr[\mcZ(\secparam) = 1] - \Pr[\mcEZ(\secparam,\epsilon) = 1]| < \epsilon,
    \]
    except for finitely many $\secpar$.
\end{enumerate}
\end{definition}

For the case of laconic CDS, the functionality $\mc{F}$ is just the \emph{conditional disclosure} functionality $\mc{F}_{cds}$: $\mc{F}_{cds}(\witness,(\statement,\mes))$ takes as input a witness $\witness$ by the receiver and a pair $(\statement,\mes)$ of statement $\statement$ and message $\mes$ by the sender. If $(\statement,\witness) \in \witrel_\lang$ it outputs $\mes$ to the receiver, otherwise $\bot$.

We will now show that context security implies the following simplified security notion for laconic CDS, which will use in our construction. In fact, this security notion is analogous to the notion of \emph{extractable witness encryption}~\cite{GKPVZ13}, which requires that any adversary which can distinguish ciphertexts must know a witness.

\begin{definition}[Witness-Knowledge Security]
Let $\lCDS = (\Setup,\Rec,\Enc,\Dec)$ be a laconic CDS scheme. Consider the following 2-phase security experiment with a two-stage adversary $\A = (\A_1,\A_2)$.
\begin{description}
\item Stage $S_1(1^\secp,\A_1)$:
\begin{itemize}
\item Compute $\crs \gets \Setup(1^\secp)$
\item $(\statement,\com,\st) \gets \A_1(1^\secp,\crs)$
\item Output $(\crs,\statement,\com,\st)$
\end{itemize} 
\item Stage $S_2(\A_2,b,\statement,\com,\st)$:
\begin{itemize}
\item Compute $\ciph^\ast \gets \Enc(\crs,\statement,\com,b)$
\item Compute $b' \gets \A_2(\st,\ciph^\ast)$
\item If $b' = b$ output 1, otherwise 0.
\end{itemize}
\end{description}

\def \Adv{\mathsf{Adv}}
\def \Ex{\mathcal{E}}

Fix the output $(\crs,\statement,\com,\st) \gets S_1(1^\secp,\A_1)$ of the first stage and assume without loss of generality that $\st$ contains $\crs,\statement,\com$. Let $\Exp_\A(\st)$ be the output of the experiment with intermediate state $\st$, i.e. $\Exp^b_\A(\st) = S_2(\A_2,b,\crs,\statement,\com,\st)$. For a given state $\st$, we define the advantage of $\A_2$ by
\[
\Adv_\st(\A_2) = \Pr[\Exp^0_\A(\st) = 1] - \Pr[\Exp^1_\A(\st) = 1].
\]

We say that an encryption scheme $(\Setup,\Enc,\Dec)$ is \emph{witness-knowledge secure} or \emph{extractable}, if there exists a PPT extractor $\Ex$, such that for every PPT adversary $\A = (\A_1,\A_2)$ the following holds. For every inverse-polynomial $\epsilon$ it holds that
\[
\Pr[\Adv_\st(\A_2) > \epsilon \text{ and } \Ex(\A_2,\st,\epsilon) \neq \sk] < \epsilon,
\]
except for finitely many $\secp$. Here the probability is taken over the random choice of $\st$. Here, we allow the runtime of $\Ex$ to be $\poly[\secp,1/\epsilon]$.

\end{definition}

We will now show that context security implies witness-knowledge security for lCDS.

\begin{theorem}
Assume that a laconic CDS scheme $\lCDS$ is context-secure. Then it is also witness-knowledge secure.
\end{theorem}

\begin{proof}
Let $\A = (\A_1,\A_2)$ be an adversary against the witness-knowledge security of $\lCDS$. Define a protocol context $\mcZ = (\mcZ_1,\mcZ_2)$ as follows. $\mcZ_1$ takes as input a crs $\crs$ and runs the first stage $S_1(1^\secp,\A_1)$ of the security experiment $(S_1,S_2)$ defined in Definition \ref{def:key_knowledge}, however using its own input $\crs$ as a common reference string instead of generating it in $S_1$. If the output of $S_2$ is $\statement,\com,\st$, $\mcZ_1$ sets $\rec = (\statement,\com)$ and outputs $(\st,\rec)$.

$\mcZ_2$ takes as input the state $\st$, chooses a random bit $b \pick \bin$ and essentially runs the second stage $S_2(\A_2,b,\crs,\statement,\com,\st)$ of the experiment in Definition \ref{def:key_knowledge}, with the difference that it does not compute $\ciph^\ast$ by itself, but uses its oracle access to $\Enc(\crs,\statement,\com,\cdot)$ to compute the challenge ciphertext $\ciph^\ast$.

By construction of $\mcZ_1$, it holds that $\mcZ_1(\crs)$ faithfully emulates the first stage of the experiment $S_1(1^\secp,\A_1)$. Moreover, for any $\crs,\statement,\com,\st$ it holds that $\mcZ_2^{\Enc(\crs,\statement,\com,\cdot)}(\st)$ faithfully simulates the second stage of the experiment $S_2(\A_2,b,\crs,\statement,\com,\st)$ for a randomly chosen bit $b \pick \bin$. Thus it follows that
\[
\Pr[\mc{Z}^{\Enc(\crs,\statement,\com,\cdot)}_2(1^\secp) = 1] = \frac{1}{2} + \frac{1}{2} \Adv_\st(\A_2).
\]
Since $\lCDS$ is context-secure, there exists a context-extractor $\ExtZ$ such that for every inverse polynomial $\epsilon$ it holds that
\begin{equation}\label{eq:context}
    |\Pr[\mcZ(\secparam) = 1] - \Pr[\mcEZ(\secparam,\epsilon) = 1]| < \epsilon.
\end{equation}

We can now define the extractor $\Ex$ for witness-knowledge security via $\Ex(\A_2,\st,\epsilon) = \ExtZ(\crs,\st,(\statement,\com),\epsilon^2/2)$. Note that we provide $\epsilon^2/2$ instead of $\epsilon$ to $\ExtZ$. Now we claim that for every inverse polynomial $\epsilon$ it holds that
\[
\Pr[\Adv_\st(\A_2) > \epsilon \text{ and } (\statement,\Ext(\A_2,\st,\epsilon)) \notin \witrel_\lang] \leq \epsilon.
\]
If this was not the case, then there exist inverse polynomial $\epsilon$ such that
\[
\Pr[\Adv_\st(\A_2) > \epsilon \text{ and } (\statement,\Ext(\A_2,\st,\epsilon)) \notin \witrel_\lang] > \epsilon
\]
for infinitely many $\secp$.

Call $\st$ \textsf{good} if $\Pr[\Adv_\st(\A_2) > \epsilon \text{ and } (\statement,\Ext(\A_2,\st,\epsilon)) \notin \witrel_\lang$. I.e. the above states that $\Pr[\st \text{ good }] > \epsilon$. Now note that if $\st$ is good, then
\[
\Pr[\mcZ(\secparam) = 1 | \st \text{ good}] = \frac{1}{2} + \frac{1}{2} \Adv_\st(\A_2) \geq \frac{1}{2} + \frac{\epsilon}{2},
\]
as $\st \text{ good}$ implies that $\Adv_\st(\A_2) > \epsilon$. Moreover, $\st \text{ good}$ also implies that $(\statement,\Ext(\A_2,\st,\epsilon)) \notin \witrel_\lang$. Consequently, if $\witness = \Ext(\A_2,\st,\epsilon)$ then the functionality $\mc{F}_{cds}(\witness,(\statement,\cdot))$ always outputs $\bot$, and therefore challenge ciphertext $\ciph^\ast$ in the experiment $\mcEZ$ is independent of the bit $b$. It follows that
\[
\Pr[\mcEZ(1^\secp,\epsilon^2/2) = 1 | \st \text{ good}] = \frac{1}{2}.
\]
Thus, it follows that
\begin{align*}
|\Pr[\mcZ(\secparam) = 1] &- \Pr[\mcEZ(\secparam,\epsilon^2/2) = 1]| \\
&\geq \underbrace{|\Pr[\mcZ(\secparam) = 1  | \st \text{ good}] - \Pr[\mcEZ(\secparam,\epsilon^2/2) = 1 | \st \text{ good}]|}_{\geq \epsilon / 2} \underbrace{\Pr[\st \text{ good}]}_{\geq \epsilon}\\
&\geq \epsilon^2 / 2.
\end{align*}
This however is in contradiction to \eqref{eq:context}. We conclude that for every inverse polynomial $\epsilon$ it holds that
\[
\Pr[\Adv_\st(\A_2) > \epsilon \text{ and } (\statement,\Ext(\A_2,\st,\epsilon)) \notin \witrel_\lang] \leq \epsilon,
\]
which shows that $\lCDS$ is witness-knowledge secure.
\end{proof}

\subsection{Big-Key Encryption from Malicious Laconic Conditional Disclosure of Secrets}

\def \PKE{\ms{PKE}}

We will now provide a construction of a big-key encryption scheme from laconic CDS. The basic idea is simple: Let $\Hash_\key: \bin^\ell \to \bin^\secp$ be a collision-resistant hash function and consider the language $\lang = \{ (\key,h) \in \bin^\secp \ | \ \exists z \in \bin^\ell \text{ s.t. } h = \Hash_\key(z) \}$ with the witness relation $\witrel_\lang = \{ ((\key,h),z) \ | \ h = \Hash_\key(z) \}$. Now let $\lCDS = (\Setup,\Rec,\Enc,\Dec)$ be a laconic CDS for the witness relation $\witrel_\lang$. The big-key encryption scheme $\PKE = (\KeyGen,\Enc,\Dec)$ is given as follows.

\begin{itemize}
    \item $\KeyGen(1^\secp,\ell)$: Choose $\key \pick \bin^\secp$, $z \pick \bin^\ell$ and set $h \gets \Hash_\key(z)$. Compute $\crs \gets \lCDS.\Setup(1^\secp)$, $(\com,\st) \gets \lCDS.\Rec(\crs,(\key,h),z)$ and output $\pk \gets (\crs,\key,h,\com)$ and $\sk \gets (\pk,\st)$
    
    \item $\Enc(\pk = (\crs,\key,h,\com),\mes)$: Compute and output $\ciph \gets \lCDS.\Enc(\crs,(\key,h),\com,\mes)$.
    
    \item $\Dec(\sk = ((\crs,\key,h,\com),z),\ciph)$: Compute and output $\mes \gets \lCDS.\Dec(\crs,(\key,h),z,\ciph)$
\end{itemize}

Correctness of this scheme follows immediately from the correctness of $\lCDS$. Moreover, note that by the compactness of $\lCDS$ we have that $|\pk|$ and $|\ciph|$ are $\poly$ but independent of the size parameter $\ell$.

We will now show that $\PKE$ is key-knowledge secure, given that $\lCDS$ is witness-knowledge secure and the hash function $\Hash$ is collision resistant.

\begin{theorem}
Assume that $\lCDS$ satisfies witness-knowledge security and $\Hash$ is collision-resistant. Then $\PKE$ is key-knowledge secure.
\end{theorem}

\begin{proof}
Let $\A = (\A_1,\A_2)$ be an adversary against $\PKE$ with inverse-polynomial advantage $\epsilon$. 

By an averaging argument, we can fix $\key \in \bin^\secp$ and $z \in \bin^\ell$ such that $\Adv(\A) > \epsilon$.

Thus, we get that for the statement $\statement = (\key,\Hash_\key(z))$ the adversary $\A$ has advantage $\epsilon$ against $\lCDS$. By the key-knowledge security of $\lCDS$ there exists an extractor $\Ext$ such that $z' \gets \Ext(\A,\epsilon)$ is a valid witness for $\statement$, except with probability $\epsilon$ over the choice of $\crs$ and $\com$. We claim that $z' = z$, except with negligible probability, which establishes that $\Ext$ is a key-extractor for $\A$.

To see this, assume that $z' \neq z$ with non-negligible probability $\epsilon'$. We can then use $\Ext(\A,\epsilon)$ construct a collision-finding adversary $\mathcal{B}$ against the hash function $\Hash$ as follows:
\begin{itemize}
    \item Input a hashing key $\key$
    \item Choose $z \pick \bin^\ell$ uniformly at random and set $h \gets \Hash_\key(z)$
    \item Compute $\crs \gets \lCDS.\Setup(1^\secp)$
    \item Compute $(\com,\st) \gets \lCDS.\Rec(\crs,(\key,h),z)$
    \item Compute $\st \gets \A_1(\crs,\st)$
    \item Compute $z' \gets \Ext(\A_2(\st),\epsilon)$
    \item Output $z,z'$
\end{itemize}

First notice that $\mathcal{B}$ is a PPT machine as $\A_1(\crs,\st)$ and $\Ext(\A_2(\st),\epsilon)$ are PPT. Observe that from the view of $\A$, $\mathcal{B}$ simulates the ciphertext indistinguishability experiment faithfully. Consequently, if $\Ext(\A,\epsilon)$ outputs a valid witness $z' \neq z$ with non-negligible $\epsilon'$, which contradicts the collision-resistance of $\Hash$.

Thus, we have established that
\[
\Pr[\Adv(\A_2) > \epsilon \text{ and } \Ext(\A_2,\epsilon) \neq \sk] < \negl,
\]
which concludes the proof.
\end{proof}

\section{PAC-Learning}

\def \learner{\mathcal{L}}
\def \classifier{\mathcal{C}}
\def \Gen{\mathsf{Gen}}
\def \sampler{\mathsf{Samp}}
\def \key{\mathsf{K}}

\def \model{\mathcal{h}}

\def \instance{\mathsf{x}}
\def \lab{\mathcal{l}}

\def \Adv{\mathsf{Adv}}

We will first fix some syntax and notation relating to the PAC-model. A learning task $\prob$ consists of the following objects.
\begin{itemize}
\item A set $X$ called the instance space. In our setting $X$ will canonically be a set of binary strings of fixed length.
\item A set $C$ of classes. In our setting we will always have $C = \bin$.
\item A problem generator algorithm $\Gen(1^\secp)$, which is a randomized algorithm which takes as input a parameter $1^\secp$ and generates a private state $\st$.
\item An algorithm $\sampler$ called instance sampler. $\sampler_\st(c)$ is indexed by a private state $\st$, takes as input a class-identifier $c \in C$ and ouputs a sample $\point \in X$.
\end{itemize}

The goal of a learning task is, given a list of labeled samples of the form $(c,\sampler_\st(c))$ to \emph{train} an efficient classifier which identifies instances with classes. We formalize the process of learning and classifying via the following two algorithms $\learner$ and $\classifier$.

\begin{itemize}
\item The learning algorithm $\learner$, takes as input a list of $m$ labeled samples $(c_1,\point_1),\dots,(c_m,\point_m)$ and computes a model/hypothesis $\model$. We will also write $\learner^{\sampler_\st(\cdot)}$ to denote that $\learner$ is given access to an unbounded number of samples via oracle access to $\sampler_\st(\cdot)$.
\item The classification algorithm $\classifier$ receives as input a model $\model$ and an instance $\point$ and outputs a class $c \in C$.
\end{itemize}

Fix a secret state $\st$ of the learning task $\prob$. We define the advantage of a classifier $\classifier$ with hypothesis $\model$ by
\[
\Adv_{\prob,\st}(\classifier_\model) = \Pr[\classifier_\model(\sampler_\st(c)) = c] - \frac{1}{|C|},
\]
where the probability is taken over the the random coins of $\sampler$ and the random choice of $c \pick C$. The advantage of a classifier measures how much better it performs on average compared to just blindly guessing the class of a given instance. 

We can now define the PAC model.
\begin{definition}[The PAC Model]
Let $\prob = (X,C,\Gen,\sampler)$ be a learning task and let $\epsilon,\delta > 0$. We say that $\prob$ is efficiently $(\epsilon,\delta)$-PAC-learnable, if there exist PPT algorithms $\learner$ and $\classifier$ such that the following holds:
\[
\Pr[ \Adv_{\prob,\st}(\classifier_\model) \geq \epsilon ] \geq  1 - \delta,
\]
where $\st \gets \Gen(1^\secp)$ and  $\model \gets \learner^{\sampler_\st(\cdot)}$. Here, the probability is taken over the random coins of $\Gen$, $\learner$ and the oracle $\sampler_\st(\cdot)$.
\end{definition}

We will only consider problems with two classes, i.e. $C = \bin$. Thus the expression for the advantage simplifies to 
\[
\Adv_{\prob,\st}(\classifier_\model) = \Pr[\classifier_\model(\sampler_\st(b)) = b] - \frac{1}{2},
\]
where the probability is taken over the the random coins of $\sampler$ and the random choice of $b \pick \bin$.

In terms of efficiency, we are interested in learning algorithms which produce models $\model$ of minimal size. 
Albeit there is no direct requirement for minimal model size, 
a growing body of works in ML focuses on obtaining small models ~\cite{crowley2018pruning,frankle2018lottery,han2015learning,luo2017thinet}. More specifically, this can also be seen as a non-triviality requirement for the learning algorithm in that the trivial strategy of just storing its input in the model $\model$ fails at this requirement.

\subsection{Robust Learning}

We will now consider learning under adversarial examples, that is we consider how well a classifier performs on inputs that are perturbed by an adversary. A perturbation adversary is an algorithm $\A$ which takes as input an instance $\point \in X$ and outputs a \emph{perturbed instance} $\tilde{\point} \in X$. Moreover, we will provide oracle-access to a classifier $\classifier_\model$ to $\A$ to model that $\A$ can test the behavior of $\classifier_\model$ on adversarial examples. Robust classification is clearly impossible against adversaries which are allowed to tamper arbitrarily. Hence, to provide a meaningful definition we need to constrain the adversary. This is typically achieved by giving the adversary a \emph{perturbation budget} specified by a \emph{metric} on the instance space $X$. Let $\metric$ be a metric on $X$. We say that an adversary $\A$ has budget $B$, if it holds for all $\point \in X$ that $\metric(\point,\A(\point)) < B$. For simplicity in the following, we will always normalize the adversary's budget to 1. This can always be achieved by rescaling the metric.

We will consider different flavors of robustness. A \emph{strongly robust classifier} will not lose its advantage, even if it receives adversarial examples as input. Fix a secret state $\st$ and a model $\model$. For a perturbation adversary $\A$, define the advantage under adversarial action as
\[
\Adv_{\A,\st}(\classifier_\model) = \Pr[\classifier_\model(\A^{\classifier_\model(\cdot)}(\st,b,\sampler_\st(b))) = c] - \frac{1}{2},
\]
where the probability is taken over the random choice of $b \pick \bin$, the random coins of $\sampler_\st$ and the random coins of $\A$.

\begin{definition}[Strongly Robust PAC Learning]
Let $\prob = (X,C,\Gen,\sampler)$ be a learning task. We say that $\prob$ is strongly robustly $(\epsilon,\delta)$-PAC-learnable in a metric $\metric$, if there exist PPT algorithms $\learner$ and $\classifier$ such that it holds for every $\metric$-constrained PPT adversary $\A$ that
\[
\Pr[\Adv_{\A,\st}(\classifier_\model) \geq \epsilon] \geq 1 - \delta
\]
where the probability is taken over the random coins of $\Gen$, $\learner$ and the oracle $\sampler_\st(\cdot)$.
\end{definition}

Note that in this definition, an adversary already wins if it diminishes the advantage of the classifier $\classifier_\model$. That is, the adversary does not necessarily need to always fool the classifier. We will now define a notion we call weak robustness which essentially requires that a successful adversary must fool the classifier $\classifier_\model$ into producing the opposite output. 

\begin{definition}[Weakly Robust PAC-Learning]\label{def:wrpac}
Let $\prob = (X,C,\Gen,\sampler)$ be a learning task. We say that $\prob$ is weakly robustly $(\epsilon,\delta,\gamma,\eta)$-PAC-learnable, if there exist PPT algorithms $\learner$ and $\classifier$ such that
\begin{enumerate}
\item $(\learner,\classifier)$ is a $(\epsilon,\delta)$-PAC learner for $\prob$
\item It holds for every PPT adversary $\A$ (with oracle access to $\classifier_\model$) and all $b \in \bin$ that
\[
| \Pr[\classifier_\model(\sampler_\key(b)) = b] - \Pr[\classifier_\model(\A^{\classifier_\model(\cdot)}(\sampler_\st(1-b))) = b] | > \gamma,
\]
except with probability $\eta$ over the choice of $\st \gets \Gen(1^\secp)$ and  $\model \gets \learner^{\sampler_\st(\cdot)}$.
\end{enumerate}
\end{definition}

We will typically require $\eta$ to be negligible, and then omit mentioning it. Condition 2 in Definition \ref{def:wrpac} essentially requires that a weakly robust classifier $\classifier_\model$ distinguishes adversarial examples from well-formed samples with advantage $\gamma$. Conversely, an adversary $\A$ fools a classifier $\classifier_\model$ if adversarial examples for class $1-b$ producesd by $\A$ are indistinguishable from well formed samples of class $b$ for $\classifier_\model$.

\subsection{Simplified Learning}

We will now consider a setting of \emph{simplified learning}, where the learning algorithm receives the secret state $\st$ as input instead of getting access to samples of $\sampler_\st(\cdot)$. As the name suggests, in the simplified setting the task of the learning algorithm is made easier as it could now just simulate a sample oracle $\sampler_\st(\cdot)$. However, recall that our goal is to construct a learning problem $\prob$ for which no small-size classifier can classify robustly in an adaptively chosen target metric. Thus, by making the job of the learning algorithm easier this simplification will only strengthen our results.

We will now show a generic transformation which transforms a learning problem $\prob$ in the simplified setting into a PAC learnable problem $\prob'$ while only slightly increasing the size of the samples $\point$. The idea is to append small shares of the secret state $\st$ to the samples. Given sufficiently many shares, the learning algorithm can reconstruct the secret state $\st$ and use a learning algorithm in the simplified model.

\def \field{\mathbb{F}}

Let $\prob$ be a classification task with generation algorithm $\Gen$ and sampler $\sampler$. Assume that $\Gen$ outputs a state $\st$ of size $\ell$. Let $\field$ be a finite field of size $2^\secp$, and $t = \ell / \secp$. Consider the following problem $\prob'$, which has the same generation algorithm $\Gen$ but uses the following sampler $\sample'$. 

\begin{description}
\item[$\sampler'_\st()$]: Run $\point \gets \sampler_\st()$. Write $\st$ as $s = (s_1,\dots,s_t) \in \field^t$. Choose a uniformly random $z \pick \field$, compute $\gamma \gets \sum_{i = 1}^t s_i z^{i - 1}$. Output $\point' \gets (\point,z,\gamma)$.
\end{description}

First note that the samples $\point'$ of $\prob'$ are of size $|\point| + 2 \secp$ and therefore small. We will now show that $\prob'$ is PAC-learnable, given that $\prob$ admits simplified learning.

Moreover, it follows immediately that any robust classifier for $\prob$ implies a robust classifier for $\prob'$. On the other hand, if $\prob$ is not robustly learnable, then neither is $\prob'$. 

\begin{theorem}
Assume that $\prob$ admits simplified robust learning. Then $\prob'$ robustly PAC-learnable.
\end{theorem}

\begin{proof}
Let $(\learner,\classifier)$ be a pair of robust learners and classifiers for $\prob$. We will construct $(\learner',\classifier')$ as follows.
\begin{itemize}
    \item $\learner'$: Query $t$ samples $(\point'_i,b_i)$ of $\prob$, where $\point'_i = (\point_i,z_i,\gamma_i)$. Interpolate a polynomial $f(X) = \sum_{i = 1}^t s'_i X^{i - 1}$ such that $f(z_i) = \gamma_i$. Parse $(s'_1,\dots,s'_t) = \st$. Compute and output $\model \gets \learner(\st)$.
    
    \item $\classifier'_\model(\point')$: Parse $\point' = (\point,z,\gamma)$, compute and output $b' \gets \classifier_\model(\point)$.
\end{itemize}

We will now briefly argue that $(\learner',\classifier')$ is a pair of robust learner and classifier.

First note that the $z_i$ are all distinct, except with probability $t \cdot 2^{-\secp}$, which is negligible. Consequently, the $(z_i,\gamma_i)$ uniquely specify the polynomial $f(X) = \sum_{i = 1}^t s_i X^{i-1}$ and it holds for all $i \in [t]$ that $s'_i = s_i$.

The claim now follows as $(\learner,\classifier)$ are a pair of robust learner and classifier for $\prob$.
\end{proof}

\subsection{Definition of the Learning Task}

We will now provide a task which is not robustly learnable if the target metric is not known ahead of time. We will start with a high-level description of the task. The problem is parametrized by a vector of keys for a big-key encryption scheme. We will refer to these keys as \emph{feature keys}. A sample of this task consists of a vector of ciphertexts, each one encrypting the identifier of the class.

It follows straightforwardly that a single feature key is sufficient to classify this task non-robustly. That is, the key $\sk_i$ allows to decrypt the ciphertext $\ciph_i$, yielding the class $b$. However, learning a feature keys is \emph{costly} as they are large in size, i.e. storing the key $\sk_i$ requires $\paramEll$ bits of storage.

Turning to robust classification, we will define our metrics in a way that allows the adversary to manipulate exactly $k$ out of the $\feaNum$ features. More specifically, the metric is indexed by a set $I \subseteq [\feaNum]$. The adversary will be allowed to arbitrarily manipulate features with index in $I$, whereas all features with index in $[\feaNum] \backslash I$ are \emph{protected}, that is the adversary is not allowed to manipulate such features.

It is thus sufficient to have a single feature key with index outside of $I$ and be aware of this fact to classify robustly. Consequently, if the learning algorithm is aware of the target metric, there is a simple robust classifier.

On the other hand, we will show that if the target metric can be chosen adaptively depending on the classifier, then there is an attack which fools the classifier with high probability.

The idea of this attack is that the adversary will be able to learn which keys the classifier knows by just having black box access to the classifier. Thus, the adversary can then choose the target metric in such a way that none of the features for which the classifier knows the keys are protected.

\begin{definition}
Let $\PKE = (\KeyGen,\Enc,\Dec)$ be a big-key encryption scheme and let $n,\ell$ be integers and $\secp$ be a security parameter. The problem $\prob_{n,\ell}$ is defined by the following algorithms $(\Gen,\sampler)$.
\begin{itemize}
	\item $\Gen(1^\secp)$: For $i = 1,\dots,\feaNum$ generate keys $(\pk_i,\sk_i) \gets \KeyGen(1^\secp,\ell)$ and output a state $\st \gets ((\pk_1,\sk_1),\dots,(\pk_\feaNum,\sk_\feaNum))$.
	\item $\sampler_\st(b \in \bin)$: For $i = 1,\dots,\feaNum$ compute $\ciph_i \gets \Enc(\pk_i,b)$. Output feature vector $(\ciph_1,\dots,\ciph_\feaNum)$.
\end{itemize}
\end{definition}

\subsection{A strongly robust small-size classifier in one target metric}

We will first provide a learning algorithm $\learner$ and a classifier $\classifier$ which robustly classifies problem $\prob{\feaNum,\ell}$ in a given target metric $\metric_\targetMSet$ which is explicitly provided to the learning algorithm via $\targetMSet$. The learning algorithm $\learner$ is provided in the simplified model in which it receives the private state $\st$ generated by $\Gen$ as input. In this construction, the model will have size $\ell + \poly$ and thus be small.

\begin{description}
\item[Learning Algorithm $\learner(\targetMSet,\st)$]:
\begin{itemize}
	\item Parse $\st \gets ((\pk_1,\sk_1),\dots,(\pk_\feaNum,\sk_\feaNum))$
	\item Fix an index $i^\ast \in \targetMSet$
	\item Set $\model \gets (i^\ast,\sk_{i^\ast})$ and output $\model$.
\end{itemize}

\item[Classifier $\classifier_\model(\point)$]:
\begin{itemize}
	\item Parse $\model = (i^\ast,\sk_{i^\ast},\ciph_{i^\ast})$
	\item Parse $\point = (\ciph_1,\dots,\ciph_\feaNum)$
	\item Compute and output $b \gets \Dec(\sk_{i^\ast})$
\end{itemize}

\end{description}

First note that the model $\model$ is of size $\ell + \log(\feaNum)$ and is therefore small.

We will briefly argue that $(\learner,\classifier)$ robustly classifies $\prob_{\feaNum,\ell}$. First assume that $\point = (\ciph_1,\dots,\ciph_\feaNum)$ is a sample of $\prob_{n,\ell}$ generated by $\sampler_\st(b)$. Then each $\ciph_i$ is of the form $\ciph_i = \Enc(\pk_i,b)$. Consequently, $\ciph_{i^\ast} = \Enc(\pk_{i^\ast},b)$ and by correctness of $\PKE$ it follows that $\Dec(\sk_{i^\ast},\ciph_{i^\ast}) = b$ and we get that $\classifier_\model(\point)$ outputs the correct class $b$.

Now let $\tilde{\point}$ be an adversarial example such that $\metric_{\targetMSet}(\point,\tilde{\point}) < 1$. Write $\tilde{\point} = (\tilde{\ciph}_1,\dots,\tilde{\ciph}_\feaNum)$. Recall that the features with index $i \in \targetMSet$ are protected, that is if $\metric_{\targetMSet}(\point,\tilde{\point}) < 1$ it must hold for all $i \in \targetMSet$ that $\tilde{\ciph}_i = \ciph_i$. But since the index $i^\ast$ is chosen such that $i^\ast \in \targetMSet$, it holds that $\tilde{\ciph}_{i^\ast} = \ciph_{i^\ast}$. Consequently, by the correctness of $\PKE$ we again get that $\Dec(\sk_{i^\ast},\ciph_{i^\ast}) = b$ and we get that $\classifier_\model(\point)$ outputs the correct class $b$.

Thus, we conclude that $\classifier_\model$ robustly classifies $\prob$ in the target-metric $\metric_\targetMSet$.

\subsection{A strongly robust large-size classifier any supported metric}

We will now show that there exists a learning algorithm and a large-size classifier for the problem $\prob_{\feaNum,\ell}$ which is robust in any supported metric $\metric_\targetMSet$. This demonstrates that robust classification of $\prob_{\feaNum,\ell}$ for an after-the-fact chosen metric is well-defined. That is, choosing the metric after the fact does not make it impossible to classify robustly, but this comes at the cost of a large description size of the classifier.

Our learning algorithm is again provided in the simplified model. Our learning algorithm and classifier follow the naive strategy: The learning algorithm learns the keys for all features whereas the classifier decrypts all features and recovers the class $b$ by making a majority decision.

\begin{description}
\item[Learning Algorithm $\learner'(\st)$]:
\begin{itemize}
	\item Parse $\st \gets ((\pk_1,\sk_1),\dots,(\pk_\feaNum,\sk_\feaNum))$
	\item Set $\model \gets (\sk_1,\dots,\sk_\feaNum)$ and output $\model$.
\end{itemize}

\item[Classifier $\classifier'_\model(\point)$]:
\begin{itemize}
	\item Parse $\model = (\sk_1,\dots,\sk_\feaNum)$
	\item Parse $\point = (\ciph_1,\dots,\ciph_\feaNum)$
	\item For all $i \in [\feaNum]$ compute $b_i \gets \Dec(\sk_i,\ciph_i)$
	\item Set $b$ to be the majority of the $b_i$, i.e. if $\sum_{i = 1}^\feaNum b_i > \feaNum/2$ set $b = 1$, otherwise $b = 0$.
	\item Output $b$
\end{itemize}

\end{description}

Note that the model $\model$ is of size $\feaNum \cdot \ell$ and therefore large.

We will now argue that $\classifier'_\model$ robustly classifies $\prob_{\feaNum,\ell}$ in any metric $\metric_\targetMSet$ for which $|\targetMSet| > \feaNum / 2$. Thus fix a $\targetMSet$ with $|\targetMSet| > \feaNum / 2$. Let $\point =(\ciph_1,\dots,\ciph_\feaNum)$ be a sample for class $b \in \bin$ and let $\tilde{\point} = (\tilde{\ciph}_1,\dots,\tilde{\ciph}_\feaNum)$ be an adversarial example with $\metric_\targetMSet(\point,\tilde{\point}) < 1$. By the definition of $\metric_\targetMSet$, it holds for all $i \in \targetMSet$ that $\tilde{\ciph}_i = \ciph_i$. Consequently, it holds for all $i \in \targetMSet$ that $b_i = \Dec(\sk_i,\tilde{\ciph}_i) = b$. However, since $|\targetMSet| > \feaNum/2$, it holds that the majority of all $b_i$ is $b$. Consequently, $\classifier'_\model(\tilde{\point})$ outputs the correct class $b$. We conclude that $\classifier'_\model$ robustly classifies $\prob$ any after-the fact chosen admissible metric $\metric_\targetMSet$.

\subsection{Impossibility of weakly robust Learning for adaptively chosen Metric}

We will now provide an efficient perturbation adversary $\A$ which fools any size-bounded classifier for $\prob_{\feaNum,\ell}$ in an adaptively chosen target metric $\metric_\targetMSet$. The idea of this adversary is that $\A$ can detect which keys the classifier $\classifier_\model$ knows by testing whether $\classifier$ notices modifications in these indices.

\begin{theorem}
Assume that $(\learner,\classifier)$ is a pair of learner and classifier for problem $\prob_{\feaNum,\ell}$ such that $\learner$ produces models $\model$ of size at most $\feaNum/2 \cdot \ell$. Let $\epsilon > 0$ and assume that $\Adv_{\prob_{\feaNum,\ell}}(\classifier_\model) \geq \frac{1}{2} + \epsilon$, except with probability $\delta$ over the choice of $\model$. Then for any $\gamma > 0$ there exists a PPT perturbation adversary $\A$ such that
\[
| \Pr[\classifier_\model(\sampler_\key(b)) = b] - \Pr[\classifier_\model(\A^{\classifier_\model(\cdot)}(\sampler_\st(1-b))) = b] | < \gamma,
\]
where the runtime of $\A$ is $\poly[\secp,1/\gamma]$. In other words, $\prob_{\feaNum,\ell}$ is not weakly robustly learnable with models of size at most $\feaNum/2 \cdot \ell$.
\end{theorem}

The following proof uses ideas relating to distinguisher dependent simulation as in~\cite{DottlingGGM19,DottlingGHMW20}.
\begin{proof}
For a subset $J \subseteq [\feaNum]$, denote by $D_J$ the following hybrid distribution.

\begin{description}
\item[Distribution $D_J(b)$]:
\begin{itemize}
\item For $i \in [\feaNum] \backslash J$ compute $\ciph_i \gets \Enc(\pk_i,b)$
\item For $i \in J$ compute $\ciph_i \gets \Enc(\pk_i,1- b)$
\item Output $(\ciph_1,\dots,\ciph_\feaNum)$.
\end{itemize}
\end{description}
That is, on $[\feaNum] \backslash J$ the $\ciph_i$ are computed as in $\prob$ by encrypting $b$, but on $J$ the $\ciph_i$ encrypt the flipped bit $1 - b$.

For a binary random variable $X \in \bin$, we will use the shorthand '''Compute an approximation of $E[X]$ with error $\delta$ of'' for the following procedure:
\begin{itemize}
\item Generate $m = \secp / \delta^2$ samples $x_1,\dots,x_m$ of $X$
\item Compute and output $\tilde{\mu} \gets \frac{1}{m} \sum_{i = 1}^m x_i$
\end{itemize}
By the Hoeffding inequality, it immediately follows that
\[
\Pr[|\tilde{\mu} - E[X]| > \delta] \leq 2 \cdot e^{-2 \secp},
\]
i.e. $\tilde{\mu}$ infact approximates $E[X]$ with an error at most $\delta$, except with negligible probability over the random choices of the approximation procedure.

Now let $\delta > 0$ be a parameter which we will set later. The perturbation adversary $\A$ is given as follows, where $\A$ gets as input a sample $\point$ of class $1-b$ and produces an adversarial example $\tilde{\point}$ which fools $\classifier_\model$ to misclassify $\tilde{\point}$ as class $b$. The strategy of $\A$ is to identify a set $J \subseteq \feaNum$ for which $\classifier_\model$ does not know the corresponding feature keys $\sk_i$.

\begin{description}
\item[Adversary $\A^{\classifier_\model(\cdot)}(\point,1-b)$]:
\begin{itemize}
\item Set $J_0 = \emptyset$
\item For $j = 1,\dots,\feaNum$:
\begin{itemize}
\item Compute and approximation $\tilde{\mu}$ of $\classifier_\model(D_J(1-b))$ with error $\delta$.
\item Compute and approximation $\tilde{\mu}'$ of $\classifier_\model(D_{J \cup \{ i \}}(1-b))$ with error $\delta$.
\item If $|\tilde{\mu}' - \tilde{\mu}| < 3 \delta$ set $J_j \gets J_{j-1} \cup \{ j \}$ otherwise $J_j = J_{j-1}$
\end{itemize}
\item If $|J_\feaNum| < n/2$ output $\bot$.
\item Let $\point = (\ciph_1,\dots,\ciph_\feaNum)$.
\item For all $i \in [\feaNum] \backslash J_\feaNum$ set $\tilde{\ciph}_i \gets \ciph_i$
\item For all $i \in J_\feaNum$ set $\tilde{\ciph}_i \gets \Enc(\pk_i,b)$
\item Output $\tilde{\point} \gets (\tilde{\ciph}_1,\dots,\tilde{\ciph}_\feaNum)$ and set the target metric to $\metric_\targetMSet$, where $\targetMSet \gets J$.
\end{itemize}
\end{description}

Let $\point$ be a sample of class $b$ and let $(\tilde{\point},\targetMSet) \gets \A^{\classifier_\model(\cdot)}(\point)$. First note that conditioned that $\A$ does not output $\bot$ it holds that $|J| < n/2$, and thus in the metric $\metric_\targetMSet$ it holds that $\metric_\targetMSet(\point,\tilde{\point}) < 1$ by the way $\targetMSet$ is chosen.

Let $\point \pick \sampler_\st(1-b)$ and $\tilde{\point} \gets \A^{\classifier_\model(\cdot)}(\point,1-b)$. Moreover let $\point' \pick \sampler_\st(1-b)$. We will now establish that 
\[
| \Pr[\classifier_\model(\point') = b] - \Pr[\classifier_\model(\tilde{\point}) = b]| < \feaNum \cdot \delta + \negl.
\]
Since we can choose the parameter $\delta$ arbitrarily small at the expense of increasing the runtime of $\A$, choosing $\delta < \gamma / \feaNum$ and the claim of the theorem follows.

Now fix a model $\model$ and consider the following hybrid experiments.

\def \Extract{\mathsf{Extract}}

\begin{itemize}
\item $\Hyb_0$: This is the real experiment, i.e. in this experiment we compute $\classifier_\model(D_\emptyset(b))$.

\item $\Hyb_i$: 
\begin{itemize}
\item Set $J_0 = \emptyset$
\item For $j = 1,\dots,i$:
\begin{itemize}
\item Compute and approximation $\tilde{\mu}$ of $\classifier_\model(D_J(b))$ with error $\delta$.
\item Compute and approximation $\tilde{\mu}'$ of $\classifier_\model(D_{J \cup \{ i \}}(b))$ with error $\delta$.
\item If $|\tilde{\mu}' - \tilde{\mu}| < 3 \delta$ set $J_j \gets J_{j-1} \cup \{ j \}$, otherwise $J_j = J_{j-1}$.
\end{itemize}
\item Compute and output $b' \gets \classifier_\model(D_J(b))$
\end{itemize}

\item $\Hyb_{\feaNum + 1}$: Same as $\Hyb_\feaNum$, but if $|J_i| < \feaNum / 2$ output $\bot$.
\end{itemize}

First notice that in $\Hyb_{\feaNum + 1}$, the output bit $b'$ is identically distributed as that of $\classifier_\model(\A^{\classifier_\model}(\point,1-b))$. We will now establish that
\[
| \Pr[\Hyb_0 = 1] - \Pr[\Hyb_{n+1} = 1] | < \feaNum \cdot \delta + \negl,
\]
which establishes that
\[
| \Pr[\classifier_\model(\sampler_\st(b)) = b] - \Pr[\classifier_\model(\A(\sampler_\st(1 - b),1-b)) = b]| < \feaNum \cdot \delta + \negl.
\]
We will first show that for all $i \in [\feaNum]$ it holds that $|\Pr[\classifier_\model(\Hyb_{i} = b)] - \Pr[\classifier_\model(\Hyb_{i-1}) = b] | < \delta + \negl$. 

Fix an index $i$ and fix the set $J_{i-1}$ computed in the first $i-1$ iterations of the loop in $\Hyb_i$. Let $\mu$ and $\mu'$ be the two approximations computed in the $i$-th iteration of the loop. We will distinguish 2 cases.
\begin{enumerate}
\item It holds that $|\mu' - \mu| \leq 3 \delta$
\item It holds that $|\mu' - \mu| > 3 \delta$
\end{enumerate}
Recall that $\mu'$ is an approximation of $E[\classifier_\model(D_{J_{i-1} \cup \{ i \} }(b))] = \Pr[\classifier_\model(D_{J_{i-1} \cup \{ i \} }(b)) = 1]$ with error $\delta$ and $\mu$ is an approximation of $E[\classifier_\model(D_{J_{i-1}}(b))] = \Pr[\classifier_\model(D_{J_{i-1}}(b)) = 1]$ with error delta. In the first case, we can conclude that
\[
| \Pr[\classifier_\model(D_{J_i}(b)) = 1] - \Pr[\classifier_\model(D_{J_{i-1}}(b)) = 1] | \leq |\mu' - \mu| + 2 \delta < 5 \delta.
\]
As in this case $\Hyb_{i}$ computes $\classifier_\model(D_{J_i }(b))$ and $\Hyb_{i-1}$ computes $\classifier_\model(D_{J_{i-1}}(b))$, it follows that
\[
| \Pr[\Hyb_{i+1} = 1] - \Pr[\Hyb_{i} = 1]| = | \Pr[\classifier_\model(D_{J_i }(b)) = 1] - \Pr[\classifier_\model(D_{J_{i-1}}(b)) = 1] | < 5 \delta.
\]
In the second case, the index $i$ will not be included in the set $J_i$ and thus $J_i = J_{i-1}$. Consequently, in this case it holds that $\Pr[\Hyb_{i+1} = 1] = \Pr[\Hyb_{i} = 1]$. Note that in this case it holds that
\[
| \Pr[\classifier_\model(D_{J_i}(b)) = 1] - \Pr[\classifier_\model(D_{J_{i-1}}(b)) = 1] | \geq |\mu' - \mu| - 2 \delta \geq \delta.
\]

It remains to show that $| \Pr[\Hyb_{\feaNum} = 1] - \Pr[\Hyb_{\feaNum + 1}] | \leq \negl$.

First note that to simulate either $\Hyb_{\feaNum}$ or $\Hyb_{\feaNum+1}$ we only need the public keys $\pk_i$ and in particular not the large private keys $\sk_i$.

First notice that by the way we constructed the set $\targetMSet = J_\feaNum$ it holds for all $i \in [\feaNum] \backslash \targetMSet$ that
\[
| \Pr[\classifier_\model(D_{J_{i-1} \cup \{i\}}(b)) = 1] - \Pr[\classifier_\model(D_{J_{i-1}}(b)) = 1] | \geq |\mu' - \mu| - 2 \delta \geq \delta.
\]
Noting that $D_{J_{i-1} \cup \{ i \}}(b)$ and $D_{J_{i-1}}$ only differ in the $i$-th feature, we can use $\classifier_\model$ to construct a distinguisher $\mc{D}_i$ which distinguishes encryptions of $0$ from encryptions of $1$ under $\pk_i$ with advantage $\delta$ as follows.
\begin{description}
\item[Distinguisher $\mc{D}_i(\ciph^\ast)$]:
\begin{itemize}
    \item For all $j \in J_{i-1}$ compute $\ciph_j \gets \Enc(\pk_j,1-b)$
    \item For all $j \in [\feaNum] \backslash (J_{i-1} \cup \{i\})$ compute $\ciph_j \gets \Enc(\pk_j,b)$.
    \item Set $\ciph_i \gets \ciph^\ast$.
    \item Set $\point' \gets (\ciph_1,\dots,\ciph_\feaNum)$.
    \item Compute and output $b' \gets \classifier_\model(\point')$
\end{itemize}
\end{description}

Clearly, if $\ciph^\ast$ is an encryption of $b$, then the sample $\point'$ constructed by $\mc{D}_i$ is a sample of $D_{J_{i-1}}$. On the other hand, if $\ciph^\ast$ is an encryption of $1 - b$, then the sample $\point'$ constructed by $\mc{D}_i$ is a sample of $D_{J_{i-1} \cup \{ i\} }$. It follows that
\[
\Adv(\mc{D}_i) = | \Pr[\classifier_\model(D_{J_{i-1} \cup \{i\}}(b)) = 1] - \Pr[\classifier_\model(D_{J_{i-1}}(b)) = 1] | \geq \delta.
\]
Consequently, by the key-knowledge security of $\PKE$ we have an extractor $\Ex$ such that $\Ex(\mc{D}_i,\st,\delta) = \sk_i$, except with probability $\delta$ over the choice of $\st$.

Consequently, using the extractors $\Ex(\mc{D}_i,\st,\delta)$ we can extract all $\sk_i$ for $i \in [\feaNum] \backslash J_\feaNum$. I.e. we can extract the string $(\sk_i)_{i \in [\feaNum] \backslash J_\feaNum}$ from $\st$. Noting that $\st$ is a string of size at most $\feaNum/2 \cdot \ell - \secp$ and all $\sk_i$ are uniformly random bit strings of length $\ell$, this implies that $|J_\feaNum| > n/2$ by Shannon's source coding theorem (as uniformly random strings cannot be compressed).

Thus, we have that $|J_\feaNum| > n/2$, except with negligible probability, which means that $| \Pr[\Hyb_{\feaNum} = 1] - \Pr[\Hyb_{\feaNum + 1}] | \leq \negl$.

This concludes the proof.
\end{proof}

\end{document}